\documentclass[letterpaper,11pt]{article}

\usepackage[runin]{abstract}
\usepackage{geometry}
\usepackage{titlesec}
\usepackage{graphicx}
\usepackage{amsmath}
\usepackage{amsthm}
\usepackage{amsfonts}
\usepackage{amssymb}
\usepackage{empheq}
\usepackage{algorithmic}
\usepackage{algorithm}
\usepackage[authoryear]{natbib}
\usepackage{url}

\usepackage[font=sf]{caption}

\titleformat*{\section}{\Large\bfseries}
\titleformat*{\subsection}{\large\sc}
\titleformat*{\subsubsection}{\itshape}

\begin{document}

\title{{\bf Evolutionary stability implies asymptotic\\ stability under multiplicative weights }}

\author{{\large{ Ioannis Avramopoulos}}
}

\date{}

\maketitle

\thispagestyle{empty} 

\newtheorem{definition}{Definition}
\newtheorem{proposition}{Proposition}
\newtheorem{theorem}{Theorem}
\newtheorem*{theorem*}{Theorem}
\newtheorem{corollary}{Corollary}
\newtheorem{lemma}{Lemma}
\newtheorem{axiom}{Axiom}
\newtheorem{thesis}{Thesis}

\vspace*{-0.2truecm}

\begin{abstract}
We show that {\em evolutionarily stable states} in general (nonlinear) population games (which can be viewed as continuous vector fields constrained on a polytope) are asymptotically stable under a {\em multiplicative weights dynamic} (under appropriate choices of a parameter called the {\em learning rate} or {\em step size}, which we demonstrate to be crucial to achieve convergence, as otherwise even chaotic behavior is possible to manifest). Our result implies that evolutionary theories based on multiplicative weights are compatible (in principle, more general) with those based on the notion of evolutionary stability. However, our result further establishes multiplicative weights as a nonlinear programming primitive (on par with standard nonlinear programming methods) since various nonlinear optimization problems, such as finding Nash/Wardrop equilibria in nonatomic congestion games, which are well-known to be equipped with a convex potential function, and finding strict local maxima of quadratic programming problems, are special cases of the problem of computing evolutionarily stable states in nonlinear population games.
\end{abstract}

\section{Introduction}
\label{introduction}

The motivation for this paper originated in online learning theory. One of the concerns in learning theory is how algorithms fare in an {\em adversarial} environment. Performance is captured by an algorithm's {\em regret}, which, given a set of possible actions, is the difference between the cost of the algorithm and the cost of the best action in hindsight. Hedge \citep{FreundSchapire1, FreundSchapire2}, for example, a well-known online learning algorithm having found applications in a variety of practical settings, is a multiplicative-weights algorithm for dynamically allocating an amount to a set of options over a sequence of steps that generalizes the weighted-majority algorithm of~\cite{Littlestone}. Freund and Schapire show that by tuning a parameter of the algorithm, called the {\em learning rate}, the algorithm adapts well to an adversarial pattern of costs in that regret vanishes with the number of steps. What if though the pattern is non-adversarial? 

In this paper, we study how multiplicative weights algorithms (such as Hedge) fare in {\em evolutionary settings} that are general enough to even capture notions of {\em nonlinear optimization}. (Our results will establish that multiplicative weights are not only capable of computing ``evolutionarily stable states'' but may also serve as nonlinear programming primitive.) 

The archetypical evolutionary setting, originally introduced by Darwin, as that is mathematically formalized using (evolutionary) game theory, is quite general: Darwin's theory of evolution has had a profound impact not only in biology but also in the social sciences and economics in a broad range of modeling environments in these disciplines. There are two distinct methods by which evolutionary theories can be approached analytically. The first is by means of notions of {\em evolutionary stability} and the second is by means of notions of {\em evolutionary dynamics}. 

In this paper, we bring these notions closer to each other by considering the dynamic behavior of a (discrete-time) multiplicative weights dynamic, which is keenly relevant to the well-known in evolutionary theory (continuous-time) replicator dynamic, in the neighborhood of an ``evolutionarily stable state.'' Let us discuss these notions more formally, however.

\subsection{Evolutionary stability}

One of the benchmark modeling approaches of mathematical biologists and economists is the {\em evolutionarily stable strategy (ESS),} which was introduced by \cite{TheLogicOfAnimalConflict} to model population states stable under evolutionary forces. \cite{Evolution} (informally) defines an ESS to be ``strategy such that, if all the members of a population adopt it, then no mutant strategy could invade the population under the influence of natural selection.'' \cite{TheLogicOfAnimalConflict} consider a strategy to be an ESS if the genetic fitness it confers to population members is greater than the fitness conferred by any mutant strategy of sufficiently small population mass. 

Consider an infinite population of organisms that interact in pairs. Each organism is genetically programmed to act according to a certain behavior (or strategy) in the pairwise interactions. The outcome of the interaction affects the {\em biological fitness} of the organisms, as measured by the number of offspring they can bear. Let $\Sigma = \{1,\ldots,n\}$ be the, common to both players, set of possible {\em pure} behaviors (strategies) and suppose that a pair of organisms meet so that one uses strategy $i \in \Sigma$ and the other uses $j \in \Sigma$. Let $a_{ij}$ be the payoff (gain in fitness) to the first organism and let $b_{ij}$ be the payoff to the second because of the interaction. A common assumption is that the payoffs are independent of the player position so that $b_{ij} = a_{ji}$. Then we can represent the payoffs for all possible interactions as an $n \times n$ matrix $C = [c_{ij}]$. In general, organisms can use {\em mixed strategies}. Then the expected gain in fitness of an organism using mixed strategy $X$ against an organism using $Y$ is $X \cdot C Y$. With this background in mind, let us make the previous statement formal.

Consider an incumbent population of unit mass in which all organisms use behavior $X^*$ and suppose that a {\em mutation} emerges so that an $\epsilon$-fraction of the organisms switch to $Y$. Let $\mathcal{Y}_{\epsilon} = (1-\epsilon)X^*+ \epsilon Y$ be the average behavior in the mixed population. Then the expected fitness of the incumbents is $X^* \cdot C \mathcal{Y}_{\epsilon}$ and that of the mutants is $Y \cdot C \mathcal{Y}_{\epsilon}$. The incumbent strategy, $X^*$, is called an ESS if for all such $Y$ and for all sufficiently small $\epsilon$ the expected fitness of the incumbents is higher than the expected fitness of the mutants, i.e., $X \cdot C \mathcal{Y}_{\epsilon} > Y \cdot C \mathcal{Y}_{\epsilon}$. Intuitively, an ESS is a behavior that is ``successful'' against all possible mutations of sufficiently small size.

But there are alternative angles whereby evolutionary stability can be approached. The original definition, for example, stipulates that
\begin{align*}
X^* \cdot C X^* &\geq X \cdot C X^*, \mbox{ } \forall X \in \mathbb{X}, \mbox{ and }\\
X^* \cdot C X^* &= X \cdot C X^* \Rightarrow X^* \cdot C X > X \cdot C X, \mbox{ } \forall X \in \mathbb{X} \mbox{ such that } X \neq X^*,
\end{align*}
where $\mathbb{X}$ is the space of mixed strategies. To understand this definition, consider a bimorphic population almost all of whose members (say an $1 - \epsilon$ fraction) are $X^*$-strategists whereas the remaining members (the $\epsilon$ fraction) are $X$-strategists. Then the fitness of the incumbents ($X^*$-strategists) is $(1-\epsilon) X^* \cdot C X^* + \epsilon X^* \cdot C X$ whereas the fitness of the mutants ($X$-strategists) is $(1-\epsilon) X \cdot C X^* + \epsilon X \cdot C X$. Since $\epsilon$ is very small, requiring the fitness of the incumbents to be greater than that of the mutants gives the conditions in the definition. 

Another characterization of the ESS, one that is analytically convenient and that directly links the notion of evolutionary stability to more abstract notions in mathematics vis-\`a-vis the variational inequalities literature, has been furnished by \cite{HSS}. This latter characterization stipulates that a strategy $X^* \in \mathbb{X}$ is an ESS if it is {\em locally superior} (cf. \citep{Weibull}), namely, if there exists a neighborhood $O$ of $X^*$ such that $X^* \cdot CX > X \cdot C X$ for all $X \neq X^*$ in $O$. 

The previous definitions generalize to more general settings wherein payoffs are nonlinear called {\em population games} (cf. \cite{PopulationGames}). Population games are {\em vector fields} constrained on a simplex, or, more generally, on a Cartesian product of simplexes, each simplex corresponding to a population. We denote this space of possible states by $\mathbb{X}$ and the vector field by $F : \mathbb{X} \rightarrow \mathbb{R}^m$, $\mathbb{X} \subset \mathbb{R}^m$. In this setting, the payoff of $X \in \mathbb{X}$ against $Y \in \mathbb{X}$ is $X \cdot F(Y)$. $X^* \in \mathbb{X}$ is an ESS if there exists a neighborhood $O$ of $X^*$ such that for all $X \in O$ we have that $(X^* - X) \cdot F(X) > 0$.

A well-known class of population games, popularized in the computer science literature by the work of \cite{HowBadisSelfishRouting} on the price of anarchy \citep{PriceOfAnarchy}, is the class of {\em nonatomic congestion games}. \cite{StableGames} show that Wardrop equilibria, as the Nash equilibria in nonatomic congestion games are also called, are, under mild assumptions, globally evolutionarily stable in that the neighborhood $O$ (wherein Nash equilibria are superior according to the aforementioned definition) corresponds to the entire state space $\mathbb{X}$. It is worth noting that nonatomic congestion games admit a potential function, which under the same mild assumptions giving rise to the existence of a (global) ESS, is strictly convex.

The previous observation that the stationary point of a strictly convex potential function in nonatomic congestion games is an ESS suggests a close relationship between the notion of evolutionary stability and that of a strict local maximum in nonlinear optimization theory. Another example where such notions coincide is the class of {\em doubly symmetric bimatrix games,} that is symmetric bimatrix games $(C, C^T)$ wherein the payoff matrix $C$ is symmetric $(C = C^T)$. In this class of games, \cite{Hofbauer-Sigmund} show that the notion of an ESS coincides with the notion of a strict local maximum of the {\em quadratic programming problem}
\begin{align*}
\mbox{ maximize } &\frac{1}{2} X \cdot S X\\
\mbox{subject to }  & X \in \mathbb{X}
\end{align*}
where $S$ is the corresponding symmetric payoff matrix and $\mathbb{X}$ is the simplex of mixed strategies of the doubly symmetric bimatrix game $(S, S)$. \cite{Bomze} refers to programs of this form as {\em standard} quadratic programming problems; \cite{BLT} show that the general quadratic programming problem over a polytope can be cast as a standard quadratic program.

\subsection{Evolutionary dynamics}

Another benchmark modeling approach of evolutionary theorists is the {\em evolutionary dynamic}. \cite{PopulationGames} posits, for example, that evolutionary game theory is the study of ``the behavior of large populations of strategically interacting agents who occasionally reevaluate their choices in light of current payoff opportunities.'' An example of an evolutionary dynamic is the {\em replicator dynamic} of \cite{TaylorJonker}, proposed shortly after the introduction of the ESS. In a symmetric bimatrix game whose payoff matrix is $C$ the replicator dynamic is the differential equation
\begin{align}
\dot{X}_i = X_i \left( (CX)_i - X \cdot CX \right), i = 1, \ldots, n,\label{RD}
\end{align}
where $n$ is the number of pure strategies. Observe that the probability (population) mass of strategies that fare better than average grows whereas the probability mass of strategies that fare worse than average diminishes. The replicator dynamic generalizes in a natural fashion to population games with nonlinear payoffs. The notion of a dynamic is intrinsically related to the notion of an {\em algorithm}. For example, the replicator dynamic and its discrete counterpart is closely related to the {\em weighted majority algorithm} of \cite{Littlestone}. Recently, this algorithm was tied in a rather strong fashion to Darwin's theory of evolution: \cite{MWUA} view evolution as a ``coordination game'' (that is, a potential game also referred to in the literature as a {\em partnership game}) solved by a {\em multiplicative weights algorithm} (cf. \citep{AHK}).

\subsection{Our questions}

In this paper, we consider a (discrete time) {\em multiplicative weights dynamic} that is akin to Hedge \citep{FreundSchapire1, FreundSchapire2}, and a model of evolution proceeding in rounds. In the setting of symmetric bimatrix games (that is, linear single-population games), in each round, the multiplicative weights dynamic selects a, generally mixed, strategy in a symmetric bimatrix game whose payoff matrix is, say $C$, then the payoffs of individual pure strategies are determined according to $C$, and the dynamic reinforces the population (probability) mass on pure strategies that fared well in the previous round. Formally, this process is characterized by the map $MW: \mathbb{X} \rightarrow \mathbb{X}$ where $\mathbb{X}$ is the space of mixed strategies of the symmetric bimatrix game $(C, C^T)$ (a simplex),
\begin{align}
MW_i(X) = X(i) \cdot \frac{1 + \alpha E_i \cdot CX}{1 + \alpha X \cdot CX}, i =1, \ldots, n,\label{intro}
\end{align}
$E_i$ is a probability vector whose mass is concentrated in position $i$, $n$ is the number of pure strategies, and $\alpha > 0$ is a parameter that we refer to as the {\em learning rate} or the {\em step size}. We note that previous equation (typically known as the {\em discrete-time replicator dynamic}) appears in a preliminary form in \citep{TaylorJonker} and later studied by \cite{TP} and \cite{Nachbar}.

From a mathematical biology perspective, \eqref{intro} can be viewed as a model of {\em asexual reproduction} in a species where generations do not overlap (for example, see \citep{Evolution, Weibull} for a derivation of \eqref{intro} in this setting). This model is, thus, more elementary than the model of evolution by {\em sexual reproduction} stipulated by \cite{MWUA} (and further explored by \cite{MPP}) wherein the multiplicative update rule is applied in each player position, however, it is also, in some sense, more general,  since we assume a general payoff structure rather than the more specific payoff structure of a partnership (doubly symmetric) game.

Depending on context, the multiplicative weights dynamic \eqref{intro} can assume similar formalizations, admitting interpretations quite disparate from those in biology. In the interest of further motivating our discrete evolution rule (vis-\'a-vis its continuous, more extensively studied, counterpart), let us, for example, consider the model of nonatomic congestion games, which is also more commonly known in the computer science literature as {\em selfish routing}. Selfish routing games take place on a directed graph whose arcs correspond to congestible resources, for example, roads in a transportation network or communication links in a computer network, and whose vertices correspond to either points of intersection, in the former case, or routers, in the latter. 

Focusing on communication networks such as the Internet, an important question is how to distribute network traffic from sources to destinations along network paths. In this setting, multiplicative weights may serve as a primitive for a routing update decision rule in the interest of load balancing traffic to reduce communication delays. Such decisions are, by design, revised in discrete steps, and, therefore, discrete dynamics are more natural than their continuous counterparts. With this application in mind, the following question arises naturally: Will adaptive routing decisions converge to equilibrium to avoid the bad performance typically associated with unpredictable oscillatory behavior? The answer, not surprisingly, depends on how the step size is selected.

Focusing attention on the step size rule, our main question in this paper is: Given a population game $F : \mathbb{X} \rightarrow \mathbb{R}^m$, where $\mathbb{X}$ is, in general, a {\em simplotope} (that is, a polytope that is a Cartesian product of simplexes), and an ESS, say $X^*$, in this game, is $X^*$ asymptotically stable under a corresponding formalization of \eqref{intro} in this more general setting and under what circumstances?

There are two general distinct approaches by which such a question can be answered. The first approach involves a linearization of the corresponding dynamical system in the neighborhood of the equilibrium (ESS) and the study of the eigenvalues of the resulting Jacobian matrix; if all eigenvalues of this matrix are inside the unit disk, asymptotic stability follows whereas if at least one eigenvalue is outside the unit disk, the system is unstable. Demonstrating asymptotic stability by this method has the advantage that a geometric rate of convergence (near the equilibrium) is immediately implied, however, there are two significant disadvantages, namely, if at least one eigenvalue is on the unit circle, the test is inconclusive (provided the remaining eigenvalues are inside the unit circle), and, furthermore, this method is often analytically intractable especially as the formula describing the evolution of the dynamical system becomes more general. The second, more creative, approach involves finding an appropriate Lyapunov function, which has the advantage that is, in principle, easier to prove global asymptotic stability results using this approach.

\subsection{Our results}

Before outlining our contributions, we believe that a few observations are in order: The previous question has been answered rather comprehensively with respect to the continuous-time replicator dynamic: Every ESS in a general nonlinear continuous population game is asymptotically stable under the replicator dynamic (by an argument analogous to that employed by \cite{StableGames} in the setting of {\em stable games,} which are population games characterized by a monotonicity property that generalizes the notion of convexity from scalar functions to vector fields, using the relative entropy as a Lyapunov function), however, the linear stability of an ESS requires a certain regularity condition (namely, that the ESS is a {\em quasistrict equilibrium}; see \citep{PopulationGames}). 

Establishing linear stability without such a regularity condition would have the unexpected implication that {\bf P = NP}: In doubly symmetric bimatrix games, the ESS coincides with the notion of asymptotic stability under the replicator dynamic \cite{Hofbauer-Sigmund}. Since the linear stability of a rest point can be checked in polynomial time and since the problem of recognizing an ESS, even in doubly symmetric bimatrix games, is {\bf coNP}-complete \citep{Etessami, Nisan-ESS},\footnote{We note in passing that the problem of recognizing an ESS in (general) symmetric bimatrix games is well-known to be {\bf coNP}-complete, however, that the problem remains {\bf coNP}-complete even in doubly symmetric bimatrix games, can be inferred by carefully considering the proof in \cite{Nisan-ESS}'s derivation.} the aforementioned unexpected result would follow.

In this paper, we show that, unless the step size is carefully chosen, evolution rules using multiplicative weights may result in chaotic behavior even in simple population games. Our main result is that an ESS in a general continuous population game is asymptotically stable under the discrete-time replicator dynamic (for appropriate choices of the step size rule) using a Lyapunov function argument based on the relative entropy. Our proof is constructive in that we provide an explicit rule that achieves this objective, and further note that the analysis is considerably more involved that the continuous-time case as we cannot rely on the analytically convenient formula for the derivative of the relative entropy but instead invoke the {\em Kantorovich inequality}. 

The Kantorovich inequality is, to some extent, a well-known tool in the analysis of nonlinear optimization problems, and it was, in fact, motivated by analytical efforts related to linear programming (it is also used in the convergence analysis of gradient-based nonlinear programming algorithms in the interest of bounding the convergence rate). However, population games generalize nonlinear programming problems to the setting of {\em vector fields}. Note that any differentiable nonlinear optimization problem can be represented as a vector field by taking the gradient of the objective function, however, vector fields capture more general modeling environments. 

The Kantorovich inequality in itself suffices to prove the asymptotic stability of an ESS assuming we may use the ESS itself in the rule for selecting a step size from iteration to iteration. Although we motivate environments where such an assumption is reasonable, our analysis would perhaps not the complete unless a step size rule that is independent of the target state were provided. To construct such a rule we prove two further lemmas, one showing the (remarkable, we believe) property of discrete-time replicator dynamic that, by applying this dynamic to any population state $X$ that is not a fixed point, we obtain another state $\hat{X}$ that is a better response to $X$ than $X$ is to itself for all values of the step size, and the other lemma showing that, assuming a small enough step size, the target ESS $X^*$ is a better response to $X$ than $\hat{X}$ is to $X$ in a neighborhood of $X^*$.

A further contribution of this paper is in constructing a rigorous case that multiplicative weights may well serve as an optimization primitive: Although the possible application of multiplicative weights in nonlinear optimization has been discussed in the literature \citep{AHK}, to the extent of our knowledge, this is the first paper that substantiates this claim with a rigorous analysis. In this vein, we also discuss analogues of the notion of an ESS to optimization theory leveraging work related to the Minty variational inequality, a rather general framework in pure mathematics, whose formalization strongly relates though to the formalization of the ESS in population games.

Our results not only demonstrate the credibility of multiplicative weights updates as a modeling method in evolutionary settings, but owing to the aforementioned observations that certain classes of population games (such as congestion games and doubly symmetric games) admit potential functions, demonstrate that such updates may well serve as a {\em nonlinear programming primitive} possible on par with gradient ascent/descent and Newton methods, for example.

The importance of the main analytical approach considered in this paper based on Lyapunov functions is argued by a linear stability analysis we carry out for a special type of population game in the setting of nonatomic congestion games, where it is demonstrated that stability analysis based on linearization quickly becomes analytically intractable in the (discrete) setting that we consider. The main prior result we use in our (far from trivial) proof in this special setting (which draws on methods from matrix analysis) is the {\em Perron-Frobenious theorem}.

We note that \cite{Weissing} mentions that an interior ESS is globally asymptotically stable under the discrete-time replicator dynamic for a small enough step size, but he never published a proof of his assertion.\footnote{We would like to thank an anonymous reviewer for pointing out this fact.}
Our results, as discussed in detail in the sequel, imply (interior) global asymptotic stability of an interior ESS under this dynamic, however, we were unable to show that such an implication holds if the step size remains constant from iteration to iteration.

Finally, toward the end of this paper, we provide an intuitive interpretation of the notion of evolutionary stability (that the aforementioned Lyapunov stability results strongly tie with multiplicative weights dynamics) in the setting of Internet architecture, in particular, with respect to the evolutionary properties of {\em routing protocols} (although extending our analysis to more general settings is certainly imaginable). Since, as discussed earlier, evolutionary stability is typically motivated in the setting of biological evolution, we believe that in this way we are making a case that notions of evolutionary stability, and, therefore, also of multiplicative weights dynamics, are bound to play a significant role in the understanding of architectural phenomena in the Internet in the interest of facilitating better, stable, and more robust network design methods.

\subsection{Other related work}

The class of population games known as {\em nonatomic congestion games,} also sometimes referred to as a {\em selfish routing games}, as noted earlier, have received considerable attention in the theoretical computer science literature. These games have the distinctive property that, as discussed in more detail later, they are equipped with a convex potential function, therefore, corresponding to a very restrictive setting of the general problem we consider in this paper (although some of our analysis, strengthening our more general arguments, takes place in this setting).

There are two distinct approaches that have been taken in the literature in the analysis of dynamic evolution in selfish routing games, namely, to either consider continuous-time evolution rules or directly analyze their discrete counterparts. \cite{SelfishRoutingEvolution} show, in one of the earliest works in this area, that Nash equilibria of selfish routing games (also sometimes referred to as Wardrop equilibria) are asymptotically stable under the continuous-time replicator dynamic, which is keenly related, as discussed earlier, to multiplicative weights updates.

However, routing can be more appropriately modeled after discrete rules of evolution as routing updates naturally take place in discrete steps rather than in a continuous fashion. Although deriving differential equations from parametric discrete-time evolution rules (by taking the derivative and letting the parameter approach zero) and studying instead the differential equation is a method often employed in the literature (for example, such an approach is taken in \citep{Piliouras} in the study of multiplicative weights in atomic congestion games), there are no general results that directly tie the behavior of the continuous system to the discrete from which it is derived.

The latter issue is related to the considerable amount of work in evolutionary game theory devoted to the study of continuous-time evolution rules (for example, based on the notion of {\em revision protocols}, that is, simple rules that nonatomic agents in a population can employ to revise their strategies according the payoffs that the collective outcome of strategic decision provide). Since a variety of phenomena in nature as well as in computer systems (such as the Internet) can be more appropriately modeled after discrete evolution rules (due to the discrete nature of, for example, computer systems), we believe that it is worth considering revisiting the work having been carried out in the evolutionary game theory literature (using the algorithmic lens) studying instead discrete counterparts of continuous evolution rules, as we do in this paper.

Discrete evolution rules related to multiplicative weights (and online learning primitives) have also been considered, for example, by \cite{FastConvergence-Journal}, however, such evolution rules are customized by design to achieve good performance in selfish routing. Our analysis demonstrates that multiplicative weights may well serve as a primitive of adaptation in a wide variety of environments (an argument that intuitively corresponds with the success that multiplicative weights-based boosting algorithms in machine learning theory, such as AdaBoost, have met practical applications).

\subsection{Overview of the rest of the paper}

In Section \ref{preliminaries}, we start off with preliminaries, in particular, we discuss symmetric bimatrix games, the framework on which much of the theory of evolutionary stability has been based, but we also present generalizations of evolutionary stability concepts in the general setting of population games (closely related to the abstract variational inequality formalization). We further introduce the general form of the multiplicative weights dynamic as well as the special case of the discrete-time replicator dynamic (which is obtained as the step size approaches zero). In Section \ref{asymptotic_stability_MW} we prove that an ESS is asymptotically stable under the discrete-time replicator dynamic in symmetric bimatrix games. This proof is generalized in Section \ref{asymptotic_stability_general_case} in the setting of general (nonlinear) population games. Section \ref{addendum_optimization} considers the question of how evolutionary stability, in its general definition in nonlinear population games, relates to the notion of the strict local optimum in nonlinear programming and provides an interpretation of the discrete-time replicator dynamic as a form of gradient. Section \ref{selfish_routing} considers a particular class of population games, namely, the well-known in the computer science literature class of selfish routing games, where the possibility of chaotic behavior is demonstrated, and a geometric rate of convergence is proven in a simple special case of the selfish routing problem. In the same section, we provide ample motivation for the notion of evolutionary stability in a setting motivated by routing in the Internet and the notion of ``incremental deployability,'' a term that originated in the networking literature. Finally, we conclude the paper in Section \ref{Conclusion}.

\section{Preliminaries}
\label{preliminaries}

The ESS, as it was originally defined \citep{TheLogicOfAnimalConflict, Evolution}, is a refinement of the symmetric Nash equilibrium in symmetric bimatrix games. Let us, therefore, start off by introducing game-theoretic concepts in the setting of bimatrix games.

\subsection{Nash equilibria in bimatrix games}

A $2$-player (bimatrix) game in normal form is specified by a pair of $n \times m$ matrices $A$ and $B$, the former corresponding to the {\em row player} and the latter to the {\em column player}. If $B = A^T$, where $A^T$ is the transpose of $A$, the game is called {\em symmetric}. A {\em mixed strategy} for the row player is a probability vector $P \in \mathbb{R}^n$ and a mixed strategy for the column player is a probability vector $Q \in \mathbb{R}^m$. The {\em payoff} to the row player of $P$ against $Q$ is $P \cdot A Q$ and that to the column player is $P \cdot B Q$. Let us denote the space of probability vectors for the row player by $\mathbb{P}$ and the corresponding space for the column player by $\mathbb{Q}$. A Nash equilibrium of a $2$-player game $(A, B)$ is a pair of mixed strategies $P^*$ and $Q^*$ such that all unilateral deviations from these strategies are not profitable, that is, for all $P \in \mathbb{P}$ and $Q \in \mathbb{Q}$, we simultaneously have that
\begin{align*}
P^* \cdot AQ^* &\geq P \cdot AQ^*\\
P^* \cdot BQ^* &\geq P^* \cdot BQ.
\end{align*}
Observe that if the bimatrix game is symmetric, the second inequality is redundant. Let $(C, C^T)$ be a symmetric bimatrix game.We call a Nash equilibrium strategy, say $P^*$, {\em symmetric,} if $(P^*, P^*)$ is an equilibrium, in which case we call the equilibrium $(P^*, P^*)$ a symmetric equilibrium.

\subsection{Evolutionary stability}

We are now ready to give the formal definition of an ESS (equivalent to the original definition due to \cite{TheLogicOfAnimalConflict} and \cite{Evolution} as shown by \cite{HSS}).  

\begin{definition}
Let $(C, C^T)$ be a symmetric bimatrix game. We say that $X^* \in \mathbb{X}$ is an ESS, if 
\begin{align*}
\exists O \subseteq \mathbb{X} \mbox{ } \forall X \in O : X^* \cdot CX > X \cdot CX.
\end{align*}
Here $\mathbb{X}$ is the space of mixed strategies of $(C, C^T)$ (a simplex in $\mathbb{R}^n$ where $n$ is the number of pure strategies) and $O$ is a neighborhood of $X^*$. 
\end{definition}

The previous definition generalizes to (nonlinear) population games: Population games are {\em vector fields} constrained on a polytope (which is also sometimes referred to as a {\em simplotope} as it is the Cartesian product of simplexes). Formally, a population game is a pair  $(\mathbb{X}, F)$. $\mathbb{X}$, the game's {\em state space}, has product form, i.e., $\mathbb{X} = \mathbb{X}_1 \times \cdots \times \mathbb{X}_n$, where the $\mathbb{X}_i$'s are simplexes and $i=1,\ldots,n$ refers to a {\em player position} (or {\em population}). To each population $i$ corresponds a set $S_i$ of $m_i$ {\em pure strategies} and a {\em population mass} $\omega_i > 0$. The {\em strategy space} $\mathbb{X}_i$ of population $i$ has the form
\begin{align*}
\mathbb{X}_i = \left\{ X \in \mathbb{R}^{m_i}_+  \bigg| \sum_{j \in S_i} X_i^j  = \omega_i \right\}.
\end{align*}
We refer to the elements of $\mathbb{X}$ as {\em states}. Each state $X \in \mathbb{X}$ can be decomposed into a vector strategies, i.e., $X = (X_1,\ldots,X_n)^T$. Let $m = \sum_i m_i$. $F: \mathbb{X} \rightarrow \mathbb{R}^m$, the game's {\em payoff function}, maps $\mathbb{X}$, the game's state space, to vectors of {\em payoffs} where each position in the vector corresponds to a pure strategy. It is typically assumed that $F$ is continuous. We adopt the following convention: In single-population games, $X(i)$ will refer to the population mass on pure strategy $i$, whereas, in multi-population games, $X_i^j$ will refer to the population mass on strategy $j$ of population $i$. 

In a general population game, the Nash equilibrium and the ESS admit the following definitions (for example, we refer the reader to \citep{PopulationGames}).

\begin{definition}
Let $(\mathbb{X}, F)$ be a population game. We say that $X^* \in \mathbb{X}$ is a Nash equilibrium, if 
\begin{align*}
\forall X \in \mathbb{X} : (X^* - X) \cdot F(X^*) \geq 0.
\end{align*} 
We say that $X^* \in \mathbb{X}$ is an ESS, if 
\begin{align*}
\exists O \subseteq \mathbb{X} \mbox{ } \forall X \in O : (X^* - X) \cdot F(X) > 0.
\end{align*}
Here $O$ is a neighborhood of $X^*$. $X^*$ is called a global ESS (GESS) if $O$ coincides with $\mathbb{X}$.
\end{definition}

The definition of the Nash equilibrium in symmetric bimatrix games can be easily shown to be a special case of the aforementioned more general definition. The latter definition bears close kinship to (in particular, it is a special case of) the general notion of critical (equilibrium) points of a {\em variational inequality}. Note that an ESS is necessarily a Nash equilibrium (in fact, isolated). 

Another word of note is that in the setting of single-population games, the acronym ESS refers, as mentioned earlier, to an evolutionarily stable {\em strategy}. The term strategy is used in the sense of a mixed strategy (that is, a probability distribution over the set of pure strategies). In the multi-population setting, however, the acronym ESS means evolutionarily stable {\em state,} which corresponds to a combination of (mixed, in general) strategies, one for each population. Keeping this in mind, no ambiguity results by referring to both situations with the term ESS (as is common in literature).

We make the following notational conventions: There are two equivalent perspectives by which to view population games, namely, one where agents in the population attempt to maximize payoffs and one where they attempt to minimize costs. We find it convenient to assume both perspectives in different parts of this paper. Whenever we assume that agents have the incentive to maximize we will use $F(\cdot)$ to denote the population game's vector field and whenever to minimize we will use $c(\cdot)$. Furthermore, states will be denoted by upper case letters in the maximization perspective (e.g., $X$) and lower case letters in the minimization perspective (e.g., $x$). Note that the aforementioned definitions of Nash equilibrium and evolutionary stability can be formalized in the minimization perspective by changing the notation (from $F(\cdot)$ to $c(\cdot)$) and flipping the inequalities.

The previous (abstract) definitions are admittedly rather dry, especially from a computer science perspective. Although mathematical analysis typically benefits from a high level of abstraction, as it implies a broader application span of the results it produces, we note that, nevertheless, ample intuition on how these notions can be cast to problems of practical interest in networked systems design is provided, toward the end of the paper, in particular, in Section \ref{intuition}.

\subsection{Multiplicative weights}

The standard game-theoretic environment wherein multiplicative weights methods are applied concerns that of playing a repeated game against an opponent (cf. \citep{FreundSchapire2}). In our setting, we assume that the environment is not controlled by an opponent but rather assumes the structure of population games (whether in the general or more particular settings).

\subsubsection{The general formalization}

Considering, Hedge, one of the most general instantiations of multiplicative weights algorithms, the update rule in the general setting of population games is given by the map $X \mapsto T(X)$, in which $X \in \mathbb{X}$ is the current state of the population game $(\mathbb{X}, F)$, and $T(X)$ is the state of the population game after the map is applied, where $T(X) = (T_1(X), \ldots, T_n(X))^T$, and
\begin{align}
T_i^j(X) = X_i^j \frac{\exp\{ \alpha F_i^j(X) \}}{(1/\omega_i) \sum_{j \in S_i} X_i^j \exp\{ \alpha F_i^j(X) \}}, \mbox{ where } i = 1, \ldots, n \mbox{ and } j \in S_i.\label{Hedge}
\end{align}
Observe, by taking the Taylor expansion of the exponential in the previous expression, that if the parameter $\alpha$ (called the learning rate or step size) is sufficiently small such that second order terms in the Taylor expansion can be ignored, and assuming further a single-population game with linear payoffs (that is, assuming $F(X) = CX$) and unit mass ($\omega = 1$), then the previous expression reduces to the discrete time replicator dynamic \eqref{intro}. Until further notice we assume such a simplified version of the multiplicative weights dynamic, which we will generalize in the sequel.

\subsubsection{Some notation}

Some of our analysis is carried out in a slightly more general form than what is actually necessary to obtain our main results, to facilitate applying the respective results to the general setting. Toward carrying out this analysis, we are going to use the following notation:
\begin{align*}
MW_i(\alpha, P, Q) \doteq P(i) \cdot \frac{1 + \alpha E_i \cdot CQ}{1 + \alpha P \cdot CQ}, i =1, \ldots, n,
\end{align*}
\begin{align*}
MW(\alpha, P, Q) = (MW_1(\alpha, P, Q), \ldots, MW_n(\alpha, P, Q))^T,
\end{align*}
and
\begin{align*}
G_i(P, Q) \doteq \frac{1 + \alpha E_i \cdot CQ}{1 + \alpha P \cdot CQ}, i =1, \ldots, n.
\end{align*}
The matrix $C$ is an $n \times n$ matrix that corresponds to the payoff matrix of a general symmetric bimatrix game. We call $\alpha > 0$ the {\em learning rate} or the {\em step size} of the multiplicative update rule.

\subsubsection{Properties of multiplicative weights dynamics}

Observe that, for all $\alpha \geq 0$,
\begin{align*}
\sum_{i=1}^n P(i) \cdot \frac{1 + \alpha (CQ)_i}{1 + \alpha P \cdot CQ} = 1.
\end{align*}
This implies that the multiplicative weights evolution rule has a property usually referred to as {\em forward invariance}. We note, without proof for the moment, that symmetric Nash equilibria (and, therefore, also evolutionarily stable strategies) are necessarily rest points of the multiplicative weights map, but that the rest points of this map are not necessarily symmetric Nash equilibria, but rather Nash equilibria of ``restricted'' instances of the game (wherein some pure strategies are absent). (This assertion is proven as Theorem \ref{fixed_points} in the general multi-population setting).

\subsubsection{Relative entropy}

Let us finally note that, in our analysis of the multiplicative weights dynamic, we are going to rely on the relative entropy function between probability distributions (also called a Kullback-Leibler divergence). The relative entropy between the $n \times 1$ probability vectors $P$ and $Q$ is given by 
\begin{align*}
RE(P, Q) = \sum_{i=1}^n P(i) \ln \left( \frac{P(i)}{Q(i)} \right).
\end{align*}
Throughout this paper we assume that $0/0 \doteq 1$ and $0 \cdot \infty \doteq 0$. Under these notational conventions, $P$ and $Q$ do not need to have full support for the definition to be valid; it is simply sufficient that the support of $P$ is a subset of the support of $Q$. We further note the standard well-known property of the relative entropy that $RE(P, Q) \geq 0$ and $RE(P, Q) = 0$ if and only if $P = Q$.

\section{Symmetric bimatrix games}
\label{asymptotic_stability_MW}

We start off by briefly introducing some relevant concepts and well-known results from the theory of discrete-time dynamical systems (cf. \citep{LaSalle}).

\subsection{Preliminaries}

Consider the map $T : \mathbb{R}^n \rightarrow \mathbb{R}^n$ where $T$ is continuous and let $X^*$ be a fixed point of $T$, that is, $T(X^*) = X^*$. We say that $X^*$ is {\em stable} under $T$ if
\begin{align*}
\forall O \exists \hat{O} \forall X \in \hat{O} \forall k \geq 0 : T^k(X) \in O
\end{align*}
where $O$ and $\hat{O}$ are neighborhoods of $X^*$. If $X^*$ is not stable, we say that it is {\em unstable}. Let $V : \mathbb{R}^n \rightarrow \mathbb{R}$, and define
\begin{align*}
\dot{V}(X) \doteq V(T(X)) - V(X).
\end{align*}
Let $\mathbb{X}$ be any set in $\mathbb{R}^n$. We say that $V$ is Lyapunov function for $T$ on $\mathbb{X}$ if $V$ is continuous on $\mathbb{R}^n$ and $\dot{V}(X) \leq 0$ for all $X \in \mathbb{X}$. We say that the function $V : \mathbb{R}^n \rightarrow \mathbb{R}$ is {\em positive definite} with respect to $X^*$ if $V(X^*) = 0$ and there is a neighborhood $O$ of $X^*$ such that, for all $X \neq X^*$ in $O$, we have that $V(X) > 0$. The following proposition is known as {\em Lyapunov's stability theorem}.

\begin{proposition}
If $V$ is Lyapunov function for $T$ on some neighborhood of $X^*$ and $V$ is positive definite with respect to $X^*$, then $X^*$ is a stable fixed point of $T$.
\end{proposition}

$X^*$ is {\em asymptotically stable} under $T$ if it is stable and there exists a neighborhood $O$ of $X^*$ such that $X_0 \in O \Rightarrow \lim_{k \rightarrow \infty} T^k(X_0) = X^*$. 

\begin{proposition}
\label{asymptotic_stability}
If $V$ and $-\dot{V}$ are positive definite with respect to $X^*$, $X^*$ is asymptotically stable.
\end{proposition}

\subsection{Some lemmas}

Before proving our main result we need the following lemmas.

\begin{lemma}
\label{simplelemma}
Let $P^*, P, Q \in \mathbb{P}$, where $\mathbb{P}$ is the space of mixed strategies of a symmetric bimatrix game $(C, C^T)$, be such that $P^*$ and $Q$ are arbitrary, $\mbox{supp}(P^*) \subseteq \mbox{supp}(P)$, and let $\hat{P} = MW_\alpha(P, Q)$. Then
\begin{align*}
RE(P^*, \hat{P}) - RE(P^*, P) \leq \sum_{i=1}^n P^*(i) \frac{1}{G_i(P, Q)} -1.
\end{align*}
\end{lemma}

\begin{proof}
We have
\begin{align*}
RE(P^*, \hat{P}) - RE(P^*, P) &= \sum_{i=1}^n P^*(i) \ln \frac{P^*(i)}{\hat{P}(i)} - \sum_{i=1}^n  P^*(i) \ln \frac{P^*(i)}{P(i)}\\
                       &= \sum_{i=1}^n P^*(i)  \ln \frac{P(i)}{\hat{P}(i)}\\
                       &= \sum_{i=1}^n P^*(i)  \ln \left( \frac{1 + \alpha P \cdot CQ}{1 + \alpha (CQ)_i} \right).
\end{align*}
Now, for any $y > 0$, $\ln y \leq y-1$, and, therefore,
\begin{align*}
RE(P^*, \hat{P}) - RE(P^*, P) &\leq \sum_{i=1}^n P^*(i)  \left( \frac{1 + \alpha P \cdot CQ}{1 + \alpha (CQ)_i} - 1\right)\\
                                         &= \sum_{i=1}^n P^*(i)  \left( \frac{1 + \alpha P \cdot CQ}{1 + \alpha (CQ)_i} \right) - 1\\
                                         &= \sum_{i=1}^n P^*(i) \frac{1}{G_i(P, Q)} - 1.\qedhere
\end{align*}
\end{proof}

The proof of the next lemma relies on the {\em Kantorovich inequality} (cf. \citep{Steele}). Let $x_1 \leq \cdots \leq x_n$ be positive numbers. Furthermore, let $\lambda_1, \ldots \lambda_n \geq 0$ such that $\sum_{i=1}^n \lambda_i = 1$. Then
\begin{align*}
\left( \sum_{i=1}^n \lambda_i x_i \right) \left( \sum_{i=1}^n \lambda_i \frac{1}{x_i} \right) \leq \left( \frac{\frac{1}{2}(x_1 + x_n)}{\sqrt{x_1 x_n}} \right)^2.
\end{align*}

\begin{lemma}
\label{ready}
Let $P^*$, $P$, and $Q$ as in the previous lemma. Then
\label{Kantorovichlemma}
\begin{align*}
\sum_{i = 1}^n P^*(i) \frac{1}{G_i(P, Q)} \leq \frac{\left( 1 + \frac{1}{2} \alpha ((CQ)_{\min} + (CQ)_{\max} \right)^2}{(1 + \alpha (CQ)_{\min}) (1+ \alpha (CQ)_{\max})} \cdot \frac{1 + \alpha P \cdot CQ}{1 + \alpha P^* \cdot CQ},
\end{align*}
where $(CQ)_{\max} = \max\{ (CQ)_i | P_i > 0 \}$ and $(CQ)_{\min} = \min\{ (CQ)_i | P_i > 0 \}$.
\end{lemma}

\begin{proof}
From the Kantorovich inequality we have
\begin{align*}
\left( \sum_{i = 1}^n P^*(i) G_i(P, Q) \right) \left( \sum_{i = 1}^n P^*(i) \frac{1}{G_i(P, Q)} \right) &\leq \left( \frac{\frac{1}{2} \left( G_{\min} + G_{\max} \right)}{\sqrt{G_{\min} G_{\max}}} \right)^2.
\end{align*}
Furthermore,
\begin{align*}
\sum_{i = 1}^n P^*(i) G_i(P, Q) = \frac{1 + \alpha P^* \cdot CQ}{1 + \alpha P \cdot CQ},
\end{align*}
and
\begin{align*}
\left( \frac{\frac{1}{2} \left( G_{\min} + G_{\max} \right)}{\sqrt{G_{\min} G_{\max}}} \right)^2 &= \frac{\left( \frac{1}{2} (G_{\min} + G_{\max}) \right)^2}{G_{\min} G_{\max}}\\
&= \frac{\left( 1 + \frac{1}{2} \alpha ((CQ)_{\min} + (CQ)_{\max}) \right)^2}{(1 + \alpha (CQ)_{\min}) (1+ \alpha (CQ)_{\max})}.
\end{align*}
Therefore,
\begin{align*}
\sum_{i = 1}^n P^*(i) \frac{1}{G_i(P, Q)} \leq \frac{\left( 1 + \frac{1}{2} \alpha ((CQ)_{\min} + (CQ)_{\max}) \right)^2}{(1 + \alpha (CQ)_{\min}) (1+ \alpha (CQ)_{\max})} \cdot \frac{1 + \alpha P \cdot CQ}{1 + \alpha P^* \cdot CQ},
\end{align*}
as in the statement of the lemma.
\end{proof}

\subsection{Evolutionary stability implies asymptotic stability}

Let $(C, C^T)$ be a symmetric bimatrix game where $C$ is an $n \times n$ matrix such that $0 < C < 1$, in that all elements of $C$ are in this interval, note that the latter assumption is made without loss of generality, let $\mathbb{X}$ be the space of mixed strategies of $(C, C^T)$ (the probability simplex in $\mathbb{R}^n$), and consider the map $T: \mathbb{X} \rightarrow \mathbb{X}$ where
\begin{align}
T_i(X) = X(i) \cdot \frac{1 + \alpha E_i \cdot CX}{1 + \alpha X \cdot CX}, i =1, \ldots, n,\label{main}
\end{align}
$E_i$ is a probability vector whose mass is concentrated in position $i$, and $\alpha : \mathbb{X} \rightarrow \mathbb{R}_+$ is a function that we refer to as the {\em step size rule}. Note that by taking the derivative of $T$ with respect to $\alpha$ and letting $\alpha = 0$, we obtain on the right-hand-side the vector field of the differential equation of the continuous-time replicator dynamic in \eqref{RD}. Although it is known that an ESS is asymptotically stable under \eqref{RD}, as mentioned in the introduction, we are not aware of any general results that tie the behavior of \eqref{main} with that of \eqref{RD}. The independent study of the discrete evolution rule we consider is, therefore, necessary. In this vein, we have the following theorem.

\begin{theorem}
\label{primary_theorem}
If $X^* \in \mathbb{X}$ is an ESS of $(C, C^T)$ where $0 < C < 1$, then there exists a step size rule $\alpha$ such that $X^*$ is asymptotically stable under \eqref{main}.
\end{theorem}

\begin{proof}
Our proof is constructive in that we provide an explicit step size rule $\alpha$ such that $X^*$ is asymptotically stable under \eqref{main}. We note that the rule that we provide here depends on knowledge of $X^*$. The techniques developed in this paper suffice to lift this dependence on $X^*$, however, to avoid an unnecessary repetition we refer the reader to the next section for the corresponding analysis in the multi-population setting of nonlinear payoffs, which generalizes Theorem \ref{primary_theorem}. However, we choose to retain this analysis to illustrate certain points discussed toward the end of the proof. Our analysis is by means of the following Lyapunov function:
\begin{align*}
V(X) = RE(X^*, X) = \sum_{i=1}^n X^*(i) \ln \left( \frac{X^*(i)}{X(i)} \right).
\end{align*}
Using standard properties of the relative entropy it can be verified that $V$ is positive definite in a neighborhood of $X^*$. Furthermore, Lemmas \ref{simplelemma} and \ref{ready} imply that
\begin{align}
\dot{V}(X) \leq \frac{\left( 1 + \frac{1}{2} \alpha ((CX)_{\min} + (CX)_{\max}) \right)^2}{(1 + \alpha (CX)_{\min}) (1+ \alpha (CX)_{\max})} \cdot \frac{1 + \alpha X \cdot CX}{1 + \alpha X^* \cdot CX} - 1.\label{Lyapunov_first_difference}
\end{align}
Observe that 
\begin{align*}
\frac{\left( 1 + \frac{1}{2} \alpha ((CX)_{\min} + (CX)_{\max}) \right)^2}{(1 + \alpha (CX)_{\min}) (1+ \alpha (CX)_{\max})} > 1
\end{align*}
whereas, since $X^*$ is an ESS,
\begin{align*}
\frac{1 + \alpha X \cdot CX}{1 + \alpha X^* \cdot CX} < 1.
\end{align*}
Now let
\begin{align*}
f(\alpha) = g(\alpha) \cdot h(\alpha)
\end{align*}
where
\begin{align*}
g(\alpha) = \frac{1 + \alpha ((CX)_{\min} + (CX)_{\max}) + \frac{1}{4} \alpha^2 ((CX)_{\min} + (CX)_{\max})^2}{1 + \alpha ((CX)_{\min} + (CX)_{\max}) + \alpha^2 (CX)_{\min} (CX)_{\max}} \equiv \frac{p(\alpha)}{q(\alpha)}
\end{align*}
and
\begin{align*}
h(\alpha) = \frac{1 + \alpha X \cdot CX}{1 + \alpha X^* \cdot CX}.
\end{align*}
\eqref{Lyapunov_first_difference} implies that $f(\alpha) - 1 \geq \dot{V}(X)$ (for all $\alpha > 0$). Note that
\begin{align*}
f(0) = 1.
\end{align*}
Observe further that
\begin{align*}
\frac{d}{d \alpha} \dot{V}(X) = f'(\alpha).
\end{align*}
Furthermore,
\begin{align*}
f'(0) = g'(0) h(0) + g(0) h'(0),
\end{align*}
and since $h(0) = g(0) = 1$, we obtain that
\begin{align*}
f'(0) = g'(0) + h'(0).
\end{align*}
Moreover
\begin{align*}
g'(\alpha) = \frac{p'(\alpha) q(\alpha) - p(\alpha) q'(\alpha)}{q^2(\alpha)},
\end{align*}
and, therefore,
\begin{align*}
g'(0) = 0.
\end{align*}
Furthermore,
\begin{align*}
h'(0) = X \cdot CX - X^* \cdot CX.
\end{align*}
Thus, we obtain
\begin{align*}
f'(0) = X \cdot CX - X^* \cdot CX < 0.
\end{align*}
Noting further that since $f$ is a continuous function of $\alpha$, we may immediately conclude that is either strictly decreasing for all $\alpha >0$ or otherwise it attains a minimum in a neighborhood of $0$ at, say $\alpha^*$, such that $f(\alpha^*) < 1$. But let us investigate \eqref{Lyapunov_first_difference} more carefully.

We would like to obtain a value of $\alpha$ such that the right-hand-side of \eqref{Lyapunov_first_difference} is negative, that is,
\begin{align*}
\frac{\left( 1 + \frac{1}{2} \alpha ((CX)_{\min} + (CX)_{\max}) \right)^2}{(1 + \alpha (CX)_{\min}) (1+ \alpha (CX)_{\max})} \cdot \frac{1 + \alpha X \cdot CX}{1 + \alpha X^* \cdot CX} < 1.
\end{align*}
We may write the previous expression as 
\begin{align*}
\left( 1 + \alpha ((CX)_{\min} + (CX)_{\max}) + \frac{1}{4} \alpha^2 ((CX)_{\min} + (CX)_{\max})^2 \right) (1 + \alpha X \cdot CX) <
\end{align*}
\begin{align*}
< \left( 1 + \alpha ((CX)_{\min} + (CX)_{\max}) + \alpha^2 (CX)_{\min} (CX)_{\max} \right) (1 + \alpha X^* \cdot CX),
\end{align*}
which, since $\alpha > 0$, is equivalent to
\begin{align*}
(X \cdot CX - X^* \cdot CX) + \frac{1}{4} \alpha \left(((CX)_{\min} + (CX)_{\max})^2 - 4(CX)_{\min} (CX)_{\max} \right) +
\end{align*}
\begin{align*}
+ \alpha ((CX)_{\min} + (CX)_{\max}) (X \cdot CX - X^* \cdot CX) +
\end{align*}
\begin{align*}
+ \frac{1}{4} \alpha^2 \left(((CX)_{\min} + (CX)_{\max})^2 (X \cdot CX) - 4 (CX)_{\min} (CX)_{\max} (X^* \cdot CX) \right) < 0
\end{align*}
Observe now that since $X \cdot CX < X^* \cdot CX$, and using straight algebra, the quadratic expression on left-hand-side of the previous inequality is less than
\begin{align*}
(X \cdot CX - X^* \cdot CX) + \frac{1}{4} \alpha ((CX)_{\max} - (CX)_{\min})^2 + \frac{1}{4} \alpha^2 ((CX)_{\max} - (CX)_{\min})^2.
\end{align*}
It, therefore, suffices to consider values of $\alpha > 0$ such that the latter expression is $< 0$. Observe that the latter quadratic expression is convex, negative when $\alpha = 0$, and, since it is convex, it attains a minimum at the stationary point $\alpha = -(1/2)$. Letting $\Delta$ be the discriminant of the quadratic expression, that is,
\begin{align*}
\Delta = \frac{1}{16} ((CX)_{\max} - (CX)_{\min})^4 - (X \cdot CX - X^* \cdot CX) ((CX)_{\max} - (CX)_{\min})^2,
\end{align*}
then the step size rule that selects any $\alpha$ such that 
\begin{align}
0 < \alpha < - \frac{1}{2} + \frac{\sqrt{\Delta} }{(1/2) ((CX)_{\max} - (CX)_{\min})^2},\label{upper_bound}
\end{align}
which is guaranteed to be positive, implies asymptotic stability of discrete-time replicator dynamic \eqref{main} by Proposition \ref{asymptotic_stability} (since for all $X$ in a neighborhood of $X^*$, $\dot{V}(X) < 0$). 
\end{proof}

The step size rule in \eqref{upper_bound} requires knowledge of $X^*$. We note that there exist applications where such a rule can be helpful. For example, a question of practical interest is, given a symmetric bimatrix game $(C, C^T)$ and a strategy $Y$ in the probability simplex of this game, whether $Y$ is an ESS (for example, see \citep{Etessami, Nisan-ESS}). Unless $Y$ is asymptotically stable under the aforementioned step size rule, we can conclude that $Y$ is not an ESS. 

In another application, there are agent populations (for example, selfish Internet users) governed by evolutionary dynamics such as the discrete-time replicator dynamic (that in the setting of the Internet may correspond to algorithms that route user traffic), and engineers interested in guiding such populations toward desirable equilibrium outcomes can only manipulate these dynamics' respective parameters (rather than the demand). In these settings, knowledge of $X^*$ can be easily be assumed; for example, in selfish routing, a subject discussed in more depth later in this paper, $X^*$ can be computed by any convex programming algorithm assuming that traffic demand is known. Multiplicative weights dynamics are, owing to their good regret performance, more appealing than, for example, gradient projection algorithms in the event of network faults and failures, which they can more easily circumvent, therefore, such a setup is not unreasonable in practical applications.

Before concluding this discussion, let us look more carefully into \eqref{upper_bound}. Observe that for all $X$ in a neighborhood of $X^*$, $\Delta > 0$, which is an implication of the definition of an ESS (i.e., $X^* \cdot CX > X \cdot CX$, $\forall X$ in a neighborhood $O$). Observe further, however, that as $X$ approaches $X^*$, $\alpha$ approaches the value $0$, which a simple implication of careful inspection of \eqref{upper_bound}. Therefore, the previous formula cannot provide a (strictly) positive upper bound on $\alpha$ such that for all values $\alpha$ less than this upper bound, the potential function is guaranteed to descend, even if we assume knowledge of $X^*$. It is an interesting question whether such an upper bound exists.

\subsection{Discussion of Weissing's result}
\label{Weissings_result}

Although our promised step size rule that is independent of the ESS will not be defined and analyzed until the next section, we believe it is reasonable to discuss Weissing's result here. Let us first prove, in this vein, that an interior ESS (i.e., an ESS supported by every pure strategy) is (interior) globally asymptotically stable under the discrete-time replicator dynamic (noting that, slightly abusing terminology, we use the term {\em interior} when we refer to asymptotic stability to mean that asymptotic stability is implied starting from an interior point of every coface of the probability simplex that supports the ESS, which in the case of an interior ESS is the entire simplex). Our proof relies on the following characterization of an ESS in a symmetric bimatrix game (corresponding to the original definition of \cite{TheLogicOfAnimalConflict} and \cite{Evolution}).

\begin{proposition}
\label{ess_characterization}
Let $(C, C^T)$ be a symmetric bimatrix games. $X^* \in \mathbb{X}$, where $\mathbb{X}$ is the space of mixed strategies of $(C, C^T)$, is an ESS if and only if 
\begin{align*}
X^* \cdot C X^* &\geq X \cdot C X^*, \mbox{ } \forall X \in \mathbb{X}, \mbox{ and }\\
X^* \cdot C X^* &= X \cdot C X^* \Rightarrow X^* \cdot C X > X \cdot C X, \mbox{ } \forall X \in \mathbb{X} \mbox{ such that } X \neq X^*.
\end{align*}
\end{proposition}

A proof of Proposition \ref{ess_characterization} can, for example, be found in \citep{Weibull}. Observe now that, by an elementary property of Nash equilibria, if $X^*$ is an interior Nash equilibrium, then, for all $X \in \mathbb{X}$, $(X^* - X) \cdot CX^* = 0$, and, therefore, by Proposition \ref{ess_characterization}, if $X^*$ is an interior ESS, for all $X \in \mathbb{X}$, $(X^* - X) \cdot CX > 0$. This implies that $X^*$ is a GESS. That $X^*$ is (interior) globally asymptotically stable under the discrete-time replicator dynamic follows now by an argument analogous to the proof of Theorem \ref{primary_theorem} (starting from an interior strategy, the potential function descends). 

\cite{Weissing} claims in his Theorem 6.5 that such an ESS is a global attractor for the discrete-time replicator dynamic (assuming a, small enough, constant step size). Although, clearly, global stability is an impossibility as, starting on the boundary of the probability simplex, unsupported pure strategies cannot receive positive population mass, it remains an interesting question whether it is possible to extend our asymptotic stability analysis assuming that the step size is constant.

\section{General population games}
\label{asymptotic_stability_general_case}

In this section, we generalize Theorem \ref{primary_theorem} to general population games. The evolution rule we consider is an approximation of \eqref{Hedge} for a small enough step size:
\begin{align}
T_i^j(X) = X_i^j \frac{1 + \alpha F_i^j(X)}{1 + \alpha (1/\omega_i) X_i \cdot F_i(X) }, \mbox{ where } i = 1, \ldots, n \mbox{ and } j \in S_i.\label{main_general}
\end{align}
Note that \eqref{main_general} coincides with \eqref{main} if we let $n=1$ and $F(X) = CX$. Our analysis proceeds starting with the same pattern as in the previous section, although it ends up being more involved due to the interaction effects among populations and the derivation of a step size rule that does not depend on knowledge of the ESS. Before proceeding with the analysis, let us first identify the fixed points of \eqref{main_general} in this more general setting. We, in fact, show that the fixed points of \eqref{main_general} coincide with the fixed points of \eqref{Hedge}, which, although expected, is something intuitively pleasing.

\subsection{Fixed points of multiplicative weights dynamics}

The following theorem was promised in the prequel.

\begin{theorem}
\label{fixed_points}
Let $S_+(X_i) = \{j \in S_i | X_i^j > 0 \}$. $X = (X_1, \ldots, X_n)^T$ is a fixed point of \eqref{main_general} if and only if, for all $i=1,\ldots,n$, and, for all $j_1, j_2 \in S_+(X_i)$, $F_i^{j_1}(X) = F_i^{j_2}(X)$. 
\end{theorem}

\begin{proof}
First we show sufficiency: If for all $i$ and for all $j_1, j_2 \in S_+(X_i)$, $F_i^{j_1}(X) = F_i^{j_2}(X)$,
then $T(X) = X$.
Some of the coordinates of $X_i$
are zero and some are positive. Clearly, 
the zero coordinates will not become positive after applying the map. 
Now, notice that, for all $j \in S_+(X_i)$, $1+\alpha F_i^j(X) =1+\alpha (1/\omega_i) \sum_{k \in S_i} X_i \cdot F_i(X)$, and, therefore, $T_i(X) = X_i$, and this is true for all $i$.

Now we show necessity: 
If $X$ is a fixed point of \eqref{main_general}, 
then for all $i$ and for all $j_1, j_2 \in S_+(X_i)$, $F_i^{j_1}(X) = F_i^{j_2}(X)$.
Let $\hat{X}_i = T_i(x)$.
Because $X$ is a fixed point, $\hat{X}_i^j = X_i^j$. Therefore,
{\allowdisplaybreaks
\begin{align}
\hat{X}_i^j &= X_i^j\notag\\
\frac{X_i^j (1+\alpha F_i^j(X))}{1+\alpha (1/\omega_i)\sum_{k} X_i^k F_i^k(X)} &= X_i^j\notag\\
1+\alpha F_i^j(X) &= 1+ \alpha (1/\omega_i)\sum_{k} X_i^k F_i^k(X)\notag\\
F_i^j(X) &= (1/\omega_i)\sum_{k} X_i^k F_i^k(X)\label{eqcondition-finalstep_1}
\end{align}
}
Equation \eqref{eqcondition-finalstep_1} implies that, for all $i=1,\ldots, n$ and for all $j \in S_i$,
\begin{align*}
X_i^j > 0 \Rightarrow F_i^j(X) = c
\end{align*}
where $c$ is a constant. This completes the proof.
\end{proof}

The following theorem asserts that the general dynamic \eqref{Hedge} has the same fixed points as \eqref{main_general}.

\begin{theorem}
Let $S_+(X_i) = \{j \in S_i | X_i^j > 0 \}$. $X = (X_1, \ldots, X_n)^T$ is a fixed point of \eqref{Hedge} if and only if, for all $i=1,\ldots,n$, and, for all $j_1, j_2 \in S_+(X_i)$, $F_i^{j_1}(X) = F_i^{j_2}(X)$. 
\end{theorem}

\begin{proof}
First we show sufficiency: If for all $i$ and for all $j_1, j_2 \in S_+(X_i)$, $F_i^{j_1}(X) = F_i^{j_2}(X)$,
then $T(X) = X$.
Some of the coordinates of $X_i$
are zero and some are positive. Clearly, 
the zero coordinates will not become positive after applying the map. 
Now, notice that, for all $j \in S_+(X_i)$, $\exp\{\alpha F_i^j(X)\} = (1/\omega_i)\sum_{k \in S_i} X_i^k \exp\{\alpha F_i^k(X)\}$. Therefore, $T_i(X) = X_i$, and this is true for all $i$.

Now we show necessity: 
If $X$ is a fixed point of Hedge, 
then for all $i$ and for all $j_1, j_2 \in S_+(X_i)$, $F_i^{j_1}(X) = F_i^{j_2}(X)$.
Let $\hat{X}_i = T_i(x)$.
Because $X$ is a fixed point, $\hat{X}_i^j = X_i^j$. Therefore,
{\allowdisplaybreaks
\begin{align}
\hat{X}_i^j &= X_i^j\notag\\
\frac{X_i^j \exp{\{\alpha F_i^j(X)\}}}{(1/\omega_i)\sum_{k} X_i^k \exp{\{\alpha F_i^k(X)\}}} &= X_i^j\notag\\
1 &= X_i^j \frac{(1/\omega_i)\sum_{k} X_i^k \exp{\{\alpha F_i^k(X)\}}}{X_i^j \exp{\{\alpha F_i^j(X)\}}}\notag\\
\omega_i &= \sum_k X_i^k \exp{\{\alpha (F_i^k(X) - F_i^j(X))\}}\notag\\
\omega_i &= X_i^j + \sum_{k \neq j} X_i^k \exp{\{ \alpha (F_i^k(X) - F_i^j(X))\}}\notag\\
\omega_i &= \omega_i - \sum_{k \neq j} X_i^k + \sum_{k \neq j} X_i^k \exp{\{\alpha (F_i^k(X) - F_i^j(X))\}}\notag\\
0 &= \sum_{k \neq j} X_i^k \left(\exp{\{ \alpha (F_i^k(X) - F_i^j(X))\}} - 1 \right)\label{eqcondition-finalstep}
\end{align}
}
Equation \eqref{eqcondition-finalstep} implies that 
\[\exp{\{\alpha (F_i^k(X) - F_i^j(X))\}} = 1, X_i^k > 0,\]
and, thus,
\[F_i^k(X) = F_i^j(X), X_i^k > 0.\]
This completes the proof.
\end{proof}

The previous theorems imply that Nash equilibria (in their general definition as given earlier) are necessarily fixed points of both \eqref{Hedge} and \eqref{main_general}. This follows immediately from a characterization of Nash equilibria in population games according to which a state is a Nash equilibrium if and only if all strategies that are in use (that is, receive a (strictly) positive population mass) in each population receive the same payoff and that strategies that are not in use (their population mass is zero) receive a payoff that is not inferior to the payoff received by strategies that are in use. (Keeping this characterization in mind, the previous theorems also imply that  there exist fixed points of the aforementioned discrete dynamics that not Nash equilibria.) Since evolutionarily stable states are necessarily Nash equilibria, they are also fixed points of such dynamics. In the rest of this section, we are concerned with the stability of evolutionarily stable states.

\subsection{Evolutionary stability implies asymptotic stability}

Let $X^* \in \mathbb{X}$ be an ESS of the population game $(\mathbb{X}, F)$. We consider the discrete-time replicator dynamic mapping $X \in \mathbb{X}$ to $\hat{X} \in \mathbb{X}$ where
\begin{align*}
\hat{X}^i_j = X_i^j \frac{1 + \alpha F_i^j(X)}{1+ \alpha (1/\omega_i) X_i \cdot F_i(X)}, \mbox{ and where } i = 1, \ldots, n \mbox{ and } j \in S_i.
\end{align*}
Recall that, as noted earlier, if $\alpha$ is small enough, \eqref{Hedge} reduces to this equation. Our stability analysis of this evolution rule is based on the Lyapunov function $RE(X^*, X)$. Since our definition of relative entropy uses probability vectors we assume that $\sum_{i} \omega_i = 1$. As we argue later, this assumption is made without loss of generality. We have:
\begin{align*}
RE(X^*, \hat{X}) - RE(X^*, X) = \sum_{i=1}^n \sum_{j \in S_i} X^{*j}_i \ln\left( \frac{1+ \alpha (1/\omega_i) X_i \cdot F_i(X)}{1 + \alpha F_i^j(X)} \right)
\end{align*}
Since, for any $y > 0$, $\ln y \leq y-1$, we obtain
\begin{align*}
RE(X^*, \hat{X}) - RE(X^*, X) \leq \sum_{i=1}^n \sum_{j \in S_i} X^{*j}_i \left( \frac{1+ \alpha (1/\omega_i) X_i \cdot F_i(X)}{1 + \alpha F_i^j(X)} \right) - \sum_{i=1}^n \omega_i.
\end{align*}
Letting
\begin{align*}
G_{ij}(X) \doteq \frac{1 + \alpha F_i^j(X)}{1+ \alpha (1/\omega_i) X_i \cdot F_i(X)},
\end{align*}
the previous expression becomes
\begin{align}
RE(X^*, \hat{X}) - RE(X^*, X) &\leq \sum_{i=1}^n \sum_{j \in S_i} X^{*j}_i \frac{1}{G_{ij}(X)} - \sum_{i=1}^n \omega_i\notag\\
  &= \sum_{i=1}^n \sum_{j \in S_i} X^{*j}_i \frac{1}{G_{ij}(X)} - 1.\label{firstt_difference}
\end{align}
Our objective is to find a step size rule such that the right-hand-side of the previous inequality is negative. To that end, we use the Kantorovich inequality:
\begin{align*}
\left( \sum_{i=1}^n \sum_{j \in S_i} X^{*j}_i G_{ij}(X) \right) \left( \sum_{i=1}^n \sum_{j \in S_i} X^{*j}_i \frac{1}{G_{ij}(X)}  \right) \leq \left( \frac{\frac{1}{2} \left( G_{\min} + G_{\max} \right)}{\sqrt{G_{\min} G_{\max}}} \right)^2,
\end{align*}
where $G_{\min} = \min_{ij} G_{ij}(X)$ and $G_{\max} = \max_{ij} G_{ij}(X)$ where $i$ and $j$ vary in strategies in the support of $X$. Note that, by the arithmetic mean / geometric mean inequality,
\begin{align*}
\left( \frac{\frac{1}{2} \left( G_{\min} + G_{\max} \right)}{\sqrt{G_{\min} G_{\max}}} \right)^2 > 1.
\end{align*}
In order for the Kantorovich inequality to be useful in the analysis, we must then have
\begin{align*}
\sum_{i=1}^n \sum_{j \in S_i} X^{*j}_i G_{ij}(X) > \sum_{i=1}^n \omega_i = 1.
\end{align*}
Before proving this, a few observations are in order: We may assume, without loss of generality, that $F_i^j(X) > 0$, for all $i = 1, \ldots, n$, $j \in S_i$, and $X \in \mathbb{X}$, since adding a large enough constant to every $F_i^j(X)$ does not affect the evolutionary properties of a Nash equilibrium. Since multiplying all such entries by the same positive constant does not affect the evolutionary properties of an equilibrium either, we may similarly assume, without loss of generality that $F_i^j(X) < 1$, and, furthermore, that $\sum_{i=1}^n \omega_i = 1$. Under these assumptions, for all $i = 1, \ldots, n$, $0 < (1/\omega_i) X_i \cdot F_i(X) < 1$.

\begin{lemma}
\label{opopopop}
\begin{align*}
\sum_{i=1}^n \sum_{j \in S_i} X^{*j}_i G_{ij}(X) > \sum_{i=1}^n \omega_i = 1.
\end{align*}
\end{lemma}

\begin{proof}
We have
\begin{align}
\sum_{i=1}^n \sum_{j \in S_i} X^{*j}_i G_{ij}(X) = \sum_{i=1}^n \omega_i \frac{1 + \alpha (1/\omega_i) X^*_i \cdot F_i(X)}{1 + \alpha (1/\omega_i) X_i \cdot F_i(X)}.\label{inter}
\end{align}
Referring to \eqref{inter}, taking a Taylor expansion of $1/(1 + \alpha (1/\omega_i) X_i \cdot F_i(X))$, and assuming $\alpha$ is sufficiently small to obtain a convergent infinite series, we obtain
\begin{align*}
\sum_{i=1}^n \omega_i \frac{1 + \alpha (1/\omega_i) X^*_i \cdot F_i(X)}{1 + \alpha (1/\omega_i) X_i \cdot F_i(X)} = \sum_{i=1}^n \omega_i \left( 1 + \alpha (1/\omega_i) X^*_i \cdot F_i(X) \right) \left[ 1 - \alpha (1/\omega_i) X_i \cdot F_i(X) + \cdots \right]
\end{align*}
The right hand side of the previous expression can be written as
\begin{align*}
\sum_{i=1}^n \omega_i \left( 1 + \alpha (1/\omega_i) X^*_i \cdot F_i(X) - \alpha (1/\omega_i) X_i \cdot F_i(X) \right) + \\
\sum_{i=1}^n \omega_i \left\{ \left( \alpha (1/\omega_i) X_i \cdot F_i(X) \right)^2 - \left( \alpha (1 / \omega_i) \right)^2 (X^*_i \cdot F_i(X)) (X_i \cdot F_i(X)) \right\} + \\
\sum_{i=1}^n \omega_i \left\{ \left( \alpha (1/\omega_i) \right)^3 (X^*_i \cdot F_i(X)) (X_i \cdot F_i(X))^2  - \left( \alpha (1/\omega_i) X_i \cdot F_i(X) \right)^3 \right\} +\\
\cdots
\end{align*}
which we may rewrite as 
\begin{align*}
\sum_{i=1}^n \omega_i + \sum_{i=1}^n \omega_i \sum_{k = 0}^{\infty} (\alpha (1/\omega_i))^k (X^*_i \cdot F_i(X) - X_i \cdot F_i(X)) (-1)^k (X_i \cdot F_i(X))^k.
\end{align*}
Rearranging terms we obtain
\begin{align*}
\sum_{i=1}^n \omega_i + \sum_{i=1}^n \omega_i (X^*_i \cdot F_i(X) - X_i \cdot F_i(X)) \sum_{k = 1}^{\infty} (\alpha (1/\omega_i))^k (-1)^{k-1} (X_i \cdot F_i(X))^{k-1},
\end{align*}
and straight algebra gives
\begin{align*}
\sum_{i=1}^n \omega_i + \sum_{i=1}^n \alpha (X^*_i \cdot F_i(X) - X_i \cdot F_i(X)) \sum_{k = 1}^{\infty} (\alpha (1/\omega_i))^{k-1} (-1)^{k-1} (X_i \cdot F_i(X))^{k-1},
\end{align*}
which is equivalent to
\begin{align*}
\sum_{i=1}^n \omega_i + \sum_{i=1}^n (X^*_i \cdot F_i(X) - X_i \cdot F_i(X)) \frac{\alpha}{1 + \alpha (1/\omega_i) X_i \cdot F_i(X)}.
\end{align*}
Recalling that, for all $i = 1, \ldots, n$, $0 < (1/\omega_i) X_i \cdot F_i(X) < 1$,
\begin{align*}
\frac{\alpha}{1 + \alpha (1/\omega_i) X_i \cdot F_i(X)} > \frac{\alpha}{1+\alpha},
\end{align*}
implying that 
\begin{align*}
\sum_{i=1}^n \omega_i &+ \sum_{i=1}^n (X^*_i \cdot F_i(X) - X_i \cdot F_i(X)) \frac{\alpha}{1 + \alpha (1/\omega_i) X_i \cdot F_i(X)} >\\
  &> \sum_{i=1}^n \omega_i + \frac{\alpha}{1+\alpha} \sum_{i=1}^n (X^*_i \cdot F_i(X) - X_i \cdot F_i(X))\\
  &= \sum_{i=1}^n \omega_i + \frac{\alpha}{1+\alpha} (X^*-X) \cdot F(X)\\
  &> \sum_{i=1}^n \omega_i,
\end{align*}
since $X^*$ is an ESS. Therefore,
\begin{align*}
\sum_{i=1}^n \sum_{j \in S_i} X^{*j}_i G_{ij}(X) > \sum_{i=1}^n \omega_i
\end{align*}
as desired. 
\end{proof}

Let us now consider the expression
\begin{align*}
\left( \frac{\frac{1}{2} \left( G_{\min} + G_{\max} \right)}{\sqrt{G_{\min} G_{\max}}} \right)^2 \frac{1}{\sum_{i=1}^n \sum_{j \in S_i} X^{*j}_i G_{ij}(X)} \equiv \left( \frac{\frac{1}{2} \left( G_{\min} + G_{\max} \right)}{\sqrt{G_{\min} G_{\max}}} \right)^2 \frac{1}{\sum_{i=1}^n \omega_i \frac{1 + \alpha (1/\omega_i) X^*_i \cdot F_i(X)}{1 + \alpha (1/\omega_i) X_i \cdot F_i(X)}},
\end{align*}
noting that our goal is to find a step size rule such that this expression is $< \sum_{i=1}^n \omega_i = 1$. To that end, we proceed in a fashion similar to that used in the previous (special) setting of symmetric bimatrix games. We, therefore, let
\begin{align*}
f(\alpha) = g(\alpha) h(\alpha)
\end{align*}
where
\begin{align*}
g(\alpha) = \left( \frac{\frac{1}{2} \left( G_{\min} + G_{\max} \right)}{\sqrt{G_{\min} G_{\max}}} \right)^2
\end{align*}
and
\begin{align*}
h(\alpha) = \left( \sum_{i=1}^n \omega_i \frac{1 + \alpha (1/\omega_i) X^*_i \cdot F_i(X)}{1 + \alpha (1/\omega_i) X_i \cdot F_i(X)} \right)^{-1}.
\end{align*}
Note that 
\begin{align}
f'(0) = g'(0) h(0) + g(0) h'(0).\label{derivative}
\end{align}
Let us consider $g(\alpha)$ first. Using the notation,
\begin{align*}
G_{\min}(X) \equiv \frac{1 + \alpha F^{\min_j}_{\min_i}(X)}{1+ \alpha (1/\omega_{\min_i}) X_{\min_i} \cdot F_{\min_i}(X)} \equiv \frac{1 + \alpha F_{\min}}{1 + \alpha \bar{F}_{\min}}
\end{align*}
and
\begin{align*}
G_{\max}(X) \equiv \frac{1 + \alpha F^{\max_j}_{\max_i}(X)}{1+ \alpha (1/\omega_{\max_i}) X_{\max_i} \cdot F_{\max_i}(X)} \equiv \frac{1 + \alpha F_{\max}}{1 + \alpha \bar{F}_{\max}},
\end{align*}
we obtain the following lemma. Before stating the lemma we note that, in this paper, the $O(\cdot)$ notation refers to the standard definition in real analysis and nonlinear optimization \citep[p. 72]{Zak} rather than that typically used in the theory of algorithms. 

\begin{lemma}
\label{regarding_g}
\begin{align*}
\left( \frac{\frac{1}{2} \left( G_{\min} + G_{\max} \right)}{\sqrt{G_{\min} G_{\max}}} \right)^2 = \frac{1}{2} \left(1 + \frac{1 + \alpha (F_{\min} + F_{\max} + \bar{F}_{\min} + \bar{F}_{\max}) + O(\alpha^2)}{1 + \alpha (F_{\min} + F_{\max} + \bar{F}_{\min} + \bar{F}_{\max}) + O(\alpha^2)} \right)
\end{align*}
\end{lemma}

\begin{proof}
We have
\begin{align*}
\left( \frac{\frac{1}{2} \left( G_{\min} + G_{\max} \right)}{\sqrt{G_{\min} G_{\max}}} \right)^2 &= \frac{\left( \frac{1}{2} (G_{\min} + G_{\max}) \right)^2}{G_{\min} G_{\max}}\\
  &= \frac{1}{4} \frac{G_{\min}}{G_{\max}} + \frac{1}{4} \frac{G_{\max}}{G_{\min}} + \frac{1}{2}\\
  &= \frac{1}{4} \frac{G_{\min}^2 + G_{\max}^2}{G_{\min} G_{\max}} + \frac{1}{2}.
\end{align*}
Therefore,
\begin{align*}
\frac{G_{\min}^2 + G_{\max}^2}{G_{\min} G_{\max}} &= \frac{\left( 1 + \alpha F_{\min} \right)^2 \left( 1+ \alpha \bar{F}_{\max} \right)^2 + \left( 1 + \alpha F_{\max} \right)^2 \left( 1+ \alpha \bar{F}_{\min} \right)^2}{\left(1 + \alpha F_{\min} \right) \left(1 + \alpha F_{\max}\right) \left(1 + \alpha \bar{F}_{\min} \right) \left(1 + \alpha \bar{F}_{\max}\right)}\\
  &= \frac{2 + 2 \alpha (F_{\min} + F_{\max} + \bar{F}_{\min} + \bar{F}_{\max}) + O(\alpha^2)}{1 + \alpha (F_{\min} + F_{\max} + \bar{F}_{\min} + \bar{F}_{\max}) + O(\alpha^2)}.
\end{align*}
Straight algebra completes the proof.
\end{proof}

Lemma \ref{regarding_g} implies $g(0) = 1$ and $g'(0) = 0$. Therefore, 
\begin{align}
f'(0) = h'(0).\label{derivative_2}
\end{align} 
Recalling from the proof of Lemma \ref{opopopop} that
\begin{align}
h(\alpha) = \left( \sum_{i=1}^n \omega_i + \sum_{i=1}^n (X^*_i \cdot F_i(X) - X_i \cdot F_i(X)) \frac{\alpha}{1 + \alpha (1/\omega_i) X_i \cdot F_i(X)} \right)^{-1},\label{h_important}
\end{align}
it is easy to verify that 
\begin{align*}
h'(0) = (X - X^*) \cdot F(X) < 0,
\end{align*}
by the assumptions that $X^*$ is an ESS and $\sum_i \omega_i = 1$, and, therefore, by \eqref{derivative_2}, that $f'(0) < 0$. To complete our argument we rely on the following lemmas.

\begin{lemma}
\label{MW-better-response_multi}
Let $X \in \mathbb{X}$, $\hat{X} = T(X)$ where $T$ corresponds to \eqref{main_general}, and assume $X$ is not a fixed point of $T$. Then, for all $\alpha > 0$,
\begin{align*}
\hat{X} \cdot F(X) > X  \cdot F(X).
\end{align*}
That is, $\hat{X}$ is a better response to $X$ than $X$ itself.
\end{lemma}

\begin{proof}
The relative entropy between the probability vectors $(1/\omega_i) \hat{X}_i$ and $(1/\omega_i) X_i$ is positive:
\begin{align*}
RE((1/\omega_i)\hat{X}_i, (1/\omega_i)X_i) = \sum_{j \in S_i} \frac{1}{\omega_i} \hat{X}_i^j \ln \left( \frac{\hat{X}_i^j}{X_i^j} \right) = \sum_{j \in S_i} \frac{1}{\omega_i} \hat{X}_i^j \ln \left( \frac{1+\alpha F_i^j(X)}{1+\alpha (1/\omega_i) X_i \cdot F_i(X)} \right) > 0.
\end{align*}
Therefore,
\begin{align*}
\ln(1 + \alpha (1/\omega_i) X_i \cdot F_i(X)) < \sum_{j \in S_i} \frac{1}{\omega_i} \hat{X}_i^j \ln\left( {1 + \alpha F_i^j(X)} \right).
\end{align*}
Exponentiating both sides we get
\begin{align}
1 + \alpha (1/\omega_i) X_i \cdot F_i(X) < \exp \left\{ \sum_{j \in S_i} \frac{1}{\omega_i} \hat{X}_i^j \ln\left( {1 + \alpha F_i^j(X)} \right) \right\}.\label{qqqwwweee_}
\end{align}
Using Jensen's inequality, we obtain
\begin{align*}
\sum_{j \in S_i} \frac{1}{\omega_i} \hat{X}_i^j \ln\left( {1 + \alpha F_i^j(X)} \right) &\leq \ln \left( \sum_{j \in S_i} \frac{1}{\omega_i} \hat{X}_i^j \left( {1 + \alpha F_i^j(X)} \right) \right)\\
                                                                                    &= \ln \left( 1 + \alpha (1/\omega_i) \hat{X}_i \cdot F_i(X) \right).
\end{align*}
Therefore, exponentiating both terms we get
\begin{align}
\exp \left\{\sum_{j \in S_i} \frac{1}{\omega_i} \hat{X}_i^j \ln\left( 1 + \alpha F_i^j(X) \right) \right\} \leq 1 + \alpha (1/\omega_i) \hat{X}_i \cdot F_i(X).\label{eeewwwqqq_}
\end{align}
Combining \eqref{qqqwwweee_} and \eqref{eeewwwqqq_}, we get
\begin{align*}
1 + \alpha (1/\omega_i) X_i \cdot F_i(X) < 1 + \alpha (1/\omega_i) \hat{X}_i \cdot F_i(X)
\end{align*}
which implies that
\begin{align*}
\hat{X}_i \cdot F_i(X) > X_i \cdot F_i(X).
\end{align*}
Summing over $i = 1, \ldots, n$, completes the proof.
\end{proof}

\begin{lemma}
\label{interim_multi}
There exists $\bar{\alpha} > 0$, which may depend on $X$, such that, for all $0 < \alpha \leq \bar{\alpha}$, 
\begin{align*}
(X^* - \hat{X}) \cdot F(X) > 0,
\end{align*}
where $\hat{X} = T(X)$ where $T$ is defined according to \eqref{main_general}.
\end{lemma}

\begin{proof}
We have
{\allowdisplaybreaks
\begin{align*}
(X^* - \hat{X}) \cdot F(X) &= \sum_{i=1}^n \sum_{j \in S_i} (X^{*j}_i - \hat{X}_i^j) F_i^j(X)\\
  &= \sum_{i=1}^n \sum_{j \in S_i} \left(X^{*j}_i - X_i^j \cdot \frac{1 + \alpha F_i^j(X)}{1 + \alpha (1/\omega_i) X_i \cdot F_i(X)} \right) F_i^j(X)\\
  &= X^* \cdot F(X) - \sum_{i=1}^n \sum_{j \in S_i} X_i^j F_i^j(X) \cdot \frac{1 + \alpha F_i^j(X)}{1 + \alpha (1/\omega_i) X_i \cdot F_i(X)}\\
  &= X^* \cdot F(X) - \sum_{i=1}^n \sum_{j \in S_i} X_i^j F_i^j(X) \cdot (1 + \alpha F_i^j(X)) (1 - \alpha (1/\omega_i) X_i \cdot F_i(X) + \cdots)\\
  &= X^* \cdot F(X) - \sum_{i=1}^n \sum_{j \in S_i} X_i^j F_i^j(X) \cdot (1 + \alpha F_i^j(X) - \alpha (1/\omega_i) X_i \cdot F_i(X)) + \cdots)\\
  &= X^* \cdot F(X) - X \cdot F(X) + \alpha \left( \sum_{i=1}^n (1/\omega_i) (X_i \cdot F_i(X))^2 - \sum_{i=1}^n \sum_{j \in S_i} X_i^j (F_i^j(X))^2 \right) + O(\alpha^2)\\
  &\geq X^* \cdot F(X) - X \cdot F(X) + \alpha \left( \sum_{i=1}^n (X_i \cdot F_i(X))^2 - \sum_{i=1}^n \sum_{j \in S_i} X_i^j (F_i^j(X))^2 \right) + O(\alpha^2)\\
  &> X^* \cdot F(X) - X \cdot F(X)\\
  &> 0.
\end{align*}
}
In the previous derivation we have taken the Taylor expansion of $1/(1+\alpha (1/\omega_i X_i \cdot F_i(X))$, noting that the corresponding infinite series is convergent provided $\alpha$ is sufficiently small. Note further that the third to last inequality is an implication that $\omega_i \leq 1$ due to the aforementioned assumption that $\sum_i \omega_i = 1$ and that the second to last inequality is an implication of Jensen's inequality.
\end{proof}

The rest of the argument is summarized in the proof of the following theorem.

\begin{theorem}
\label{primary_theorem_general}
If $X^* \in \mathbb{X}$ is an ESS of the population game $(\mathbb{X}, F)$, then there exists a step size rule $\alpha$ such that $X^*$ is asymptotically stable under \eqref{main_general}.
\end{theorem}

\begin{proof}
We use $V(X) = RE(X^*, X)$ as our Lyapunov function. Therefore, $\dot{V}(X) = f(\alpha) - 1$. Recall from the proof of Lemma \ref{opopopop} that (since $\sum_i \omega_i$ = 1)
\begin{align*}
h(\alpha) < \left( 1 + (\alpha / (1+\alpha)) (X^*-X) \cdot F(X) \right)^{-1}.
\end{align*}
Letting $\hat{X}$ as in Lemma \ref{interim_multi} and using this lemma, we obtain
\begin{align*}
f(\alpha) = g(\alpha) h(\alpha) < g(\alpha) \hat{h}(\alpha) \equiv \hat{f}(\alpha),
\end{align*}
where
\begin{align*}
\hat{h}(\alpha) = \left( 1 + (\alpha / (1 + \alpha)) (\hat{X}(\alpha)-X) \cdot F(X) \right)^{-1}
\end{align*}
for all $0 < \alpha \leq \bar{\alpha}$, where $\bar{\alpha} > 0$. (In the previous expression showing the dependence of $\hat{X}$ on $\alpha$ is for emphasis.)  Lemma \ref{MW-better-response_multi} implies that $\hat{h}(\alpha) < 1$ for all $\alpha > 0$. Taking the Taylor expansions of $g(a)$ at $a=0$ and the Taylor expansion of $\hat{h}(\alpha)$, we obtain
\begin{align*}
\hat{f}(\alpha) = \left( g(0) + \alpha g'(0) + O(\alpha^2) \right) \cdot \left(1 - (\alpha/(1+\alpha)) (\hat{X}(\alpha)-X) \cdot F(X) + O(\alpha^2) \right).
\end{align*}
However, $g(0) = 1$ and $g'(0) = 0$, and, therefore,
\begin{align}
\hat{f}(\alpha) = 1 - (\alpha/(1+\alpha)) (\hat{X}(\alpha)-X) \cdot F(X) + O(\alpha^2),\label{ineq_multi}
\end{align}
implying that, for small enough $\alpha$, $\hat{f}(\alpha) < 1$, by Lemma \ref{MW-better-response_multi}. This analysis, therefore, suggests, the following step size rule: Select some initial value, say $\alpha_0$, for the step size, and compare the value $\hat{f}(\alpha_0)$ with one. If $\hat{f}(\alpha_0) < 1$, transition to the new state $\hat{X}(\alpha_0)$. Otherwise let $\alpha_1 = \alpha_0 / 2$, and repeat this process until $\hat{f}(\alpha_k) < 1$, where $\alpha_k = \alpha_0 / 2^k$ and $k$ is a positive integer. Lemma \ref{interim_multi} and \eqref{ineq_multi} ensure this process terminates in a finite number of steps. Using this step size rule (resembling {\em line search} in nonlinear optimization) and invoking Proposition \ref{asymptotic_stability} completes the proof.
\end{proof}

\subsection{Discussion on simpler step size rules}

In the rest of this section, we discuss the case of using simpler step size rules such as a constant step size in the discrete-time replicator dynamic. To that end, note that although 
\begin{align*}
g(\alpha) = \left( \frac{\frac{1}{2} \left( G_{\min} + G_{\max} \right)}{\sqrt{G_{\min} G_{\max}}} \right)^2 > 1,
\end{align*}
and referring to Lemma \ref{regarding_g}, $g(\alpha)$ can approximate the value $1$ with arbitrary accuracy for a small enough value of $\alpha$. The expression in Lemma \ref{regarding_g}, in fact, implies that $g(\alpha)$ approaches the value $1$ (from above) much faster (quadratically so) than $\alpha$ as $\alpha$ approaches $0$, which is desirable with respect to the convergence rate. Recall further that, by Lemma \ref{opopopop} and the Kantorovich inequality, the first difference of the Lyapunov function $RE(X^*, X)$, as expressed in \eqref{firstt_difference}, is negative everywhere but in small neighborhood, say $\mathcal{N}$, of $X^*$, which can be made arbitrarily small by an appropriate selection of the step size $\alpha$ (since, by Lemma \ref{regarding_g}, $g(a)$ approaches $1$ quadratically, from above, whereas $h(\alpha)$ approaches $1$, from below, only linearly as $\alpha \rightarrow 0$ for fixed $X$). 

Furthermore, since, as is well-known, the relative entropy between two probability distributions exceeds the square of the Euclidean distance between these distributions (assuming such distance is positive), the previous observation about the existence of a neighborhood $\mathcal{N}$ around $X^*$ where the first difference of the Lyapunov function is negative, implies the existence of an arbitrarily small neighborhood around $X^*$ (that is, a ball of arbitrarily small radius and center $X^*$ intersecting the evolution space $\mathbb{X}$), say $\mathcal{N}^*$, such that the discrete-time replicator dynamic is guaranteed to enter and remain thereafter even if a constant step size rule is used. We, therefore, expect that, in practice, an ESS can be approximated with arbitrary precision even by a constant step size rule. 

However, the previous analysis further implies that other step size rules, more sophisticated than using a constant step size but less demanding in terms of function evaluations than the previous one implying asymptotic stability, also suffice to ensure a property analogous to what we have just described. In this vein, since the setting we are considering is one where multiple populations interact, it is meaningful to generalize the discrete-time replicator dynamic such that each population adapts according to a population-dependent learning rate $\alpha_i$ as follows:
\begin{align*}
\hat{X}_i^j = X_i^j \frac{1 + \alpha_i F_i^j(X)}{1+ \alpha_i (1/\omega_i) X_i \cdot F_i(X)}, \mbox{ } i = 1, \ldots, n \mbox{ , } j \in S_i.
\end{align*}
Letting now 
\begin{align*}
\alpha_i = \kappa \frac{\omega_i}{X_i \cdot F_i(X)}, i = 1, \ldots, n
\end{align*}
where $\kappa$ is a small constant (which can always be chosen as such since $F$ is continuous and the domain of $F$ is a compact polytope), and recalling that $f(\alpha) = g(\alpha) h(\alpha)$ where $h(\alpha)$ and referring to either \eqref{inter} or \eqref{h_important}, we obtain (without resorting to the argument in Lemma \ref{opopopop}) that
\begin{align*}
\sum_{i=1}^n \sum_{j \in S_i} X^{*j}_i G_{ij}(X) > \sum_{i=1}^n \omega_i = 1,
\end{align*}
and, therefore, that the previous argument regarding the existence of $\mathcal{N}^*$ for the constant step size rule applies to this latter (adaptive) step size rule as well. In closing this section, we should note that we cannot preclude the case that asymptotic stability using a small enough constant step size rule may well be happening, and our inability to prove such a result could be due to the analytical approach we have taken. However, our arguments correspond, we believe, to a rather compelling case that the constant step size rule could work well in practice.

\section{Multiplicative weights and nonlinear optimization}
\label{addendum_optimization}

In this section, we relate multiplicative weights with concepts from nonlinear optimization theory. To that end, we first relate the ESS with the notion of a strict local optimum. Our primary objective, in this vein, is to further motivate the generality of our results on the asymptotic stability of the ESS under multiplicative weights dynamics (in particular, Theorem \ref{primary_theorem_general}) with the general nonlinear optimization framework. In the same vein, we subsequently relate the multiplicative weights dynamic with the gradient of a continuously differentiable objective function.

\subsection{Strict local optima and the ESS}

Population games are a quite general modeling environment. A special case of a population game is a nonlinear programming problem wherein the constraint set is a simplotope and the objective function is continuously differentiable. (This is easy to see as follows: Since the objective function is continuously differentiable, the corresponding gradient field obtained by taking the gradient of the objective function is continuous. Therefore, viewing such an optimization problem as a particular instance of a population game is nothing but an alternative means of representing the same problem.) An important solution concept in nonlinear programming is the strict local optimum (minimum or maximum), which corresponds to a point of the constraint set whose objective value is strictly lower (assuming a minimum) than the objective values of all nearby points (in a neighborhood). The notion of a strict local minimum does not coincide with the notion of an evolutionarily stable state for arbitrary objective functions; interesting counterexamples can be demonstrated through inspection of objective functions considered by \cite{Absil}. 

Considering these examples it is natural to ask to what extent the notion of evolutionarily stable states in potential games (that is, games admitting a potential function), or equivalently continuously differentiable nonlinear programming problems over a simplotope, coincide with strict local optima. \cite{StableGames} show that if the potential function is strictly convex (at least locally), then these notions coincide. Furthermore, as mentioned earlier in the introduction, in doubly symmetric bimatrix games, that is, symmetric bimatrix games $(C, C^T)$ where $C$ is a symmetric matrix ($C = C^T$), strict local maxima are known to coincide with evolutionarily stable states (or, more precisely, strategies in this special setting) as shown by \cite{Hofbauer-Sigmund}. In fact, as mentioned in the introduction, the general quadratic programming problem over a polytope can be cast in this setting. But these are not the only cases in which such coincidence between strict local optima and evolutionarily stable states happens. 

Notions of convexity of continuously differentiable scalar functions can be captured through notions of monotonicity of their corresponding gradient fields (viewed as operators). For example, convex functions give rise to monotone gradient fields. Monotone vector fields are more general than convex functions, but a monotone vector field that admits a potential function gives rise to a convex function. In fact, the aforementioned work of \cite{StableGames} on stable games was motivated by observations along these lines: Stable games generalize potential games with a concave potential function. Note that strictly stable games give rise to a GESS.

Quasimonotone operators generalize monotone operators: If a quasimonotone operator admits a potential function, the corresponding scalar field is a quasiconvex or quasiconcave function; for example, the Gaussian probability distribution is quasiconcave. \cite{Konnov} shows that strictly quasimonotone operator $F$ over a compact convex domain $\mathbb{X}$ implies the existence of a unique element $X^* \in \mathbb{X}$ such that, for all $X \in \mathbb{X}$, $(X^* - X) \cdot F(X) > 0$. If $\mathbb{X}$ is simplotope, $(\mathbb{X}, F)$ is a population game, and Konnov's result implies $X^*$ is a GESS. A straightforward implication of this result is that the stationary point of a (locally) strictly quasiconvex function is also an ESS. 

\if(0) 

In the interest of gaining intuition in the problem at hand we provide a simple  proof using elementary notions of analysis and optimization that an interior point of a locally quasiconvex function is an ESS. Although this is special case of Konnov's result, our proof is more elementary. Recall that a function is called quasi-convex, if its {\em sublevel sets} are convex, where the sublevel set of $f : \mathbb{X} \rightarrow \mathbb{R}$ at a feasible point $X$, denoted as $L_f(x)$, is
\begin{align*}
L_f(x) = \{ y \in \mathbb{X} | f(y) \leq f(x) \}.
\end{align*}
In the next proposition, we assume for simplicity that the point in the constraint set that is a strict local minimum is the origin, which avoids cluttering the proof with unnecessary notation. Recalling the definition of evolutionary stability, and noting that we assume for simplicity that $x^* \equiv 0$, the condition the following theorem implies is the minimization analogue of evolutionary stability. In our proof, we use a rather fundamental result in nonlinear optimization theory (for example, see \citep{Bertsekas}) that the gradient at an interior point of the constraint set of a nonlinear program (whose objective function is continuously differentiable) is perpendicular to the level set of that point (an analogous rather technical result also holds for boundary points of the constraint set).

\begin{proposition}
\label{laskdfhjdhfnnncmvn}
Let $f: \mathbb{X} \subseteq \mathbb{R}^m \rightarrow \mathbb{R}$ be $C^1$, $\mathbb{X}$ is a convex subset of $\mathbb{R}^m$, $\mathbb{X}$ contain the origin, let the origin be an interior strict local minimum of $f$, and suppose $f$ is locally quasi-convex near the origin. Then, there exists a neighborhood $O$ of the origin such that, for all $x \in O$, $x^T \nabla f(x) > 0$.
\end{proposition}

\begin{figure}[tb]
\centering
\includegraphics[width=14cm]{figure.eps}
\caption{\label{sldfksdkjfhdjfh}
Unless $F(z) = \nabla f(z)$ is the zero vector, it is impossible for the level set at $z$ (which is convex by assumption) to contain the entire line segment between $z$ and $y$. }
\end{figure}

\begin{proof}
Let $O$ be a neighborhood of the origin such that, for all $x \in O - \{0\}$, $f(0) < f(x)$ and such that, for all $x \in O - \{0\}$, $\nabla f(x) \neq 0$. (Such a neighborhood exists by the assumption of the theorem.) Suppose, for the sake of contradiction, that there exists $x \in O$ such that $x^T \nabla f(x) \leq 0$. Then, by the mean value theorem, there exists $\epsilon \in (0, 1)$ such that $f(0) - f(x) = - x^T \nabla f(\epsilon x)$. Since $f(0) - f(x) < 0$ by the assumption that the origin is a strict local minimum, this is equivalent to saying that there exists $y$ in between the origin and $x$ such that $y^T \nabla f(y) > 0$. Since $f$ is $C^1$, $\nabla f$ is continuous, and by the intermediate value theorem, there exists $z$ in between $y$ and $x$ such that $z^T \nabla f(z) = 0$ and such that, for all $w$ in between $y$ and $z$, $w^T \nabla f(w) > 0$. This is equivalent to saying that $f(\cdot)$ must be strictly increasing in the direction $x$ in between $y$ and $z$. Therefore, the sublevel set $S = \{ \psi \in O | f(\psi) \leq f(z) \}$ must contain the set of all such aforementioned $w$. However, $S$ is necessarily perpendicular to $\nabla f(z)$. Furthermore, $S$ is convex by the assumption that the origin is locally quasi-convex. Therefore, it must necessarily be the case that $\nabla f(z) = 0$ (see Figure~\ref{sldfksdkjfhdjfh}), which contradicts the assumption of the lemma that, for all $x \in O$, $\nabla f(x) \neq 0$, and, therefore, also the hypothesis that there exists $x \in O$ such that $x^T \nabla f(x) \leq 0$.
\end{proof}

Note that the origin may not belong to the simplotope domain of a population game; for example, the origin lies outside the probability simplex. However, a simple translation of the domain brings the problem to the setting of the previous theorem. 

\fi

In the rest of this section, we continue our investigation on nonlinear optimization.

\subsection{Multiplicative weights and gradient descent}

An elementary result in calculus states that, in the setting of unconstrained optimization, the gradient of a continuous differentiable objective function points to the direction of maximum increase of the objective function. It is, therefore, not surprising that a variety of nonlinear programming algorithms are based on the gradient. For example, in minimization problems, the iterates of {\em gradient descent algorithms} correspond to taking linear steps along the direction of the gradient (and projecting the point obtained after such a linear step to the constraint set of the corresponding problem, assuming such a point is infeasible). In the rest of this section, we show that the steps multiplicative weights dynamics take in the constraint set of nonlinear optimization problems bear close kinship to such gradient-based methods. Recalling the aforementioned notion of forward variance, which can be easily shown to hold in the general setting, multiplicative weights updates do not require a projection step. Taking on a minimization perspective, we show that multiplicative weights dynamics initially descend the objective function assuming a small enough step size. In this sense, it would not be an overstatement to say that they correspond to a form of gradient in arbitrary nonlinear programs (wherein the constraint set is a simplotope).

\subsubsection{Preliminaries}

We have insofar represented population games from a ``maximization'' perspective wherein the incentive structure is represented by a (vector-valued) payoff function. Population games are sometimes more appropriately represented considering an equivalent ``minimization'' perspective using a (vector-valued) {\em cost function} (the payoff function's negative). We are going to denote such a cost function by $c$ whose domain is the space of feasible states $\mathbb{X}$. We are going to assume the same notational convention as in payoff vectors, namely, $c_i^j(x)$ is the cost of strategy $j$ in population $i$ at state $x$. In this setting, \eqref{Hedge} becomes
\begin{align}
T_i^j(x) = x_i^j \frac{\exp\{ - \alpha c_i^j(x) \}}{(1/\omega_i) \sum_{j \in S_i} x_i^j \exp\{ - \alpha c_i^j(x) \}}, \mbox{ where } i = 1, \ldots, n \mbox{ and } j \in S_i.\label{Hedge_minimization}
\end{align}
We note in passing that viewing population games from the perspective of minimizing costs is a particularly appropriate perspective in the setting of selfish routing games we consider in the next section. In the setting of selfish routing, states correspond to network flows, populations correspond to commodities, and strategies within in population/commodity correspond to network paths. The objective is to minimize delay from source to destination, and costs are interpreted as such.

We consider a population game whose incentive structure can be represented by a continuously differentiable scalar function, say $\Phi$, whose domain is $\mathbb{X}$. As mentioned earlier, selfish routing games, for example, admit a potential function that is, in fact, convex. At present we do not assume convexity though. The algebraic expression of the potential function can be complicated, but all we care about is that its gradient equals the cost function $c$, that is, at any feasible state $x \in \mathbb{X}$, $\nabla \Phi (x) = c(x)$. Let us further recall from optimization theory the notion of a sublevel set of a scalar function: The sublevel set of $\Phi$ at a feasible state $x$, denoted as $L_\Phi(x)$, is
\begin{align*}
L_\Phi(x) = \{ y \in \mathbb{X} | \Phi(y) \leq \Phi(x) \}.
\end{align*}
If $\Phi$ is, for example, a convex or quasi-convex function, its sublevel sets are convex sets but, in general sublevel sets can take on arbitrary shapes. Note that in unconstrained optimization problems the gradient of the objective function at a point of the domain is perpendicular to the corresponding level set at that point. Since we are concerned with a constrained optimization problem, the analogous statement is that the projection of the cost vector, $c(x)$, at an interior state $x$ at the tangent space of the constraint space is perpendicular to the sublevel set $L_\Phi(x)$ at $x$.

\subsubsection{Multiplicative weights as a descent algorithm}

Considering a potential game with potential function $\Phi : \mathbb{X} \rightarrow \mathbb{R}$, the next theorem demonstrates that Hedge (as given by the update rule in \eqref{Hedge_minimization}, the minimization analogue of \eqref{Hedge}) is a descent algorithm (in the sense of nonlinear programming), assuming a sufficiently small step size, for the game's potential function. Note, in this vein, that given a potential game and starting at any state $x$ in the space of feasible states $\mathbb{X}$, \eqref{Hedge_minimization} defines a (nonlinear) curve $\mathcal{C}_x: (0, +\infty) \rightarrow \mathbb{X}$ on the space of feasible states parametrized by the learning rate $\alpha$, and, depending on the choice of $\alpha$, the next iterate is generated by moving to a point along $\mathcal{C}_x$. (Contrast this with the {\em linear step} that the majority of nonlinear programming algorithms take.) Since the entire curve $\mathcal{C}_x$ lies on the space $\mathbb{X}$ of feasible flows, Hedge does not require an extra projection step to achieve feasibility (which is a property we earlier called forward invariance). The next theorem shows that $\mathcal{C}_x(\alpha)$ defines (initially, for a small value of the step size) a descent direction in the sense that the tangent of this curve at $x$ (corresponding to $\alpha = 0$) is a direction wherein the objective function decreases locally.

\begin{theorem}
Given any interior $x \in \mathbb{X}$ that is not a fixed point, the tangent of $\mathcal{C}_x$ at $x$ is a descent direction of $\Phi$.
\end{theorem}

\begin{proof}
The tangent of $\mathcal{C}_x$ at $x$ is $(d \mathcal{C}_x / d \alpha)(0)$. The normal of the level set of $\Phi$ at $x$ is $[\nabla \Phi(x)]^+ = [c(x)]^+$ where $[\cdot]^+$ denotes the projection of $c(x)$ onto the tangent space of $\mathbb{X}$. To prove the theorem, it suffices to prove that the angle between these vectors is obtuse, i.e., $(d \mathcal{C}_x / d \alpha)(0) \cdot [c(x)]^+ < 0$. By straight calculus, we have
\begin{align*}
\frac{d \mathcal{C}_x}{d \alpha}(0) = \left[ x_i^j \left( \frac{1}{\omega_i} \sum_{k \in S_i} x^k_i c^k_i(x) - c^j_i(x) \right) \right]_{ij}.
\end{align*}
(The expression on the right hand side is a vector of size $m = \sum_i m_i$ whose element corresponding to population $i$ and strategy $j$ is given in the brackets. Note that, as a function of $x$, this vector field precisely corresponds to the vector field of the continuous-time replicator dynamic.) Furthermore, \cite{PopulationGames} shows that
\begin{align*}
[c(x)]^+ = \left[ c_i^j(x) - \mathbf{1} \left( \frac{1}{m_i} \sum_{k \in S_i} c_k^j(x) \right) \right]_{ij}
\end{align*}
where $\mathbf{1}$ is an $m \times 1$ vector of ones. Therefore,
\begin{align*}
\frac{d \mathcal{C}_x}{d \alpha}(0) \cdot c(x) &= \sum_{i=1}^n \omega_i \left( \left( \frac{1}{\omega_i} \sum_{j \in S_i} x^j_i c^j_i(x) \right)^2 - \frac{1}{\omega_i} \sum_{j \in S_i} x_i^j \left( c^j_i(x) \right)^2 \right)\\
&- \sum_{i = 1}^n \left( \frac{1}{m_i} \sum_{k \in S_i} c_k^j(x) \right) \sum_{j \in S_i} x_i^j \left( \frac{1}{\omega_i} \sum_{k \in S_i} x^k_i c^k_i(x) - c^j_i(x) \right)\\
&= \sum_{i=1}^n \omega_i \left( \left( \frac{1}{\omega_i} \sum_{j \in S_i} x^j_i c^j_i(x) \right)^2 - \frac{1}{\omega_i} \sum_{j \in S_i} x_i^j \left( c^j_i(x) \right)^2 \right) < 0
\end{align*}
by Jensen's inequality.
\end{proof}

To understand the previous theorem recall that the vector $[c(x)]^+$ at $x$ is perpendicular to the level set $L_\Phi(x)$ at $x$. The theorem then states that for a small enough value of the step size $\alpha$, the next iterate lies inside $L_\Phi(x)$, and the objective value is decreased. We should note, however, that $\alpha$ continues to increase $\mathcal{C}_x(\alpha)$ is not guaranteed to remain in the level set $L_\Phi(x)$. The importance of carefully selecting the step size is further argued in the next section.

\section{Selfish routing games}
\label{selfish_routing}

In this section, we are concerned with the setting of selfish routing games whose origins can be traced to road traffic research \citep{Wardrop}. They were popularized in the computer science literature by the work of \cite{HowBadisSelfishRouting} on the price of anarchy \citep{PriceOfAnarchy}. Selfish routing games are population games where each population corresponds to a commodity from a source node to a destination node in a directed graph whose arcs correspond to congestible resources (e.g., communication links in a computer network).

This class of games is equipped with a convex potential function, which is strictly convex under the (reasonable in many practical settings) assumption that link cost functions are strictly increasing. By an argument analogous that used in the proof of Theorem \ref{primary_theorem_general}, the (unique) stationary point of this potential function (which corresponds to the Nash/Wardrop equilibrium), is (interior) globally asymptotically stable under Hedge (and the discrete-time replicator dynamic). 

In the sequel, we take this result further. In particular, our objectives, in this vein, are twofold: First, we demonstrate an example of a simple selfish routing game in which there exists a step size such that the induced behavior of Hedge is chaotic (implying care must be taken in selecting the step size in practice). Second, we analyze a simple selfish routing game (consisting of parallel arcs between a pair of vertices) using linearization, proving that all eigenvalues of the corresponding Jacobian matrix lie inside the unit disc (in all but a number of systems that have measure zero). Our analysis implies the stronger results of asymptotic stability under a small enough constant step size as well as a geometric rate of convergence (in this particular setting), however, the example is also used to demonstrate the difficulty of employing this type of analysis (based on linearization) in the generality of the discrete setting of evolutionary dynamics considered in this paper.

In closing this section, we provide an intuitive interpretation of evolutionary stability in a practical, from a computer science perspective, problem related to network systems design and the Internet. Although the notion of evolutionary stability was originally motivated and developed from the perspective of mathematical biology, we argue that computer science and networked systems design research can, in principle, benefit by gaining closer familiarity with this concept. Selfish routing games provide an appropriate framework to thoroughly articulate such a case.

\subsection{Preliminaries}

Selfish routing games are a special case of {\em nonatomic congestion games} whenever the latter assume an interpretation corresponding to routing flow in a network represented as a (directed) graph. We may define a nonatomic congestion game as the following five-tuple $$(E, N, \{ \omega_i \}_{i \in N}, \{ S_i \}_{i \in N}, \{c_e(\cdot)\}_{e \in E}).$$ $E$ is a finite set of {\em resources} such as communication links in a network represented by a graph. $N = \{1,\ldots,n \}$ is a finite set of {\em commodities} such as source-destination pairs. To each commodity $i$ corresponds a scalar {\em demand} $\omega_i > 0$ and a set of {\em paths} (or {\em pure strategies}) $S_i \subseteq 2^E$. The pure strategies typically correspond to simple paths in the network's graph. The demand of each commodity is allocated to paths by (infinitesimal) {\em agents} (or {\em packets}) inducing {\em flow} on the resources. 

A nonatomic congestion game is a special case of a population game where $N$ corresponds to the (finite) set of $n$ populations, $\omega_i, i \in N$ is the mass of population $i$, and $S_i$ is the set of pure strategies of the same population. States corresponds to (multi-commodity, as they typically called) flows. We denote the space of all multi-commodity flows by $\mathbb{X}$. A multi-commodity flow $x \in \mathbb{X}$ is the ``aggregate behavior'' of a population of agents (which we prefer to think of as packets) of mass $\sum_{i=1}^n \omega_i$ where each packet has infinitesimal mass (hence the name {\em nonatomic}). We may view the population as being partitioned in $n$ sub-populations (one for each commodity) so that the fraction of packets assigned to sub-population $i$ is $\omega_i/(\sum_i \omega_i)$. 

In the previous definition of a nonatomic congestion game, each pure strategy is a subset of the collection of resources in $E$ (that is, the links comprising the path that the pure strategy corresponds to). Let $S_e$ be the set of pure strategies that assign positive flow to resource $e \in E$, and let $x_e$ be flow assigned to $e$, that is, $x_e = \sum_{i \in N, j \in S_e} x_i^j$. To each resource $e$ corresponds an increasing {\em cost function} $c_e(\cdot)$ that measures the cost (or {\em delay}) of the resource as a function of $x_e$. The delay $c_i^j(x)$ of path $j \in S_i$ is then given by $c_i^j(x) = \sum_{e \in j} c_e(x_e),$
and it is equal to the delay that a packet (of commodity $i$) suffers when it is sent along path $j$ under flow $x$.

\cite{Beckmann-McGuire-Winsten} show that, under the previous assumptions, selfish routing games are potential games and that the potential function is convex; if the delay function of each link is a strictly increasing function of the amount of flow in the link, then the potential function is strictly convex, in which case, \cite{StableGames} show that the critical point of the potential function is a GESS, an ESS  that is ``superior'' against every other state of the selfish routing game.

\subsection{The possibility of chaotic behavior}

If the learning rate is sufficiently large, Hedge may induce chaotic behavior.
We demonstrate this possibility through an example of a scalar system that has period three; this implies the possibility of chaotic behavior by the famous theorem that ``period three implies chaos''~\citep{PeriodThree}. Noting that Hedge acting on a system of two parallel links gives rise to a scalar dynamical system, our example consists of two links with cost functions $c_1(x) = x$ and $c_2(x) = 10 \cdot x$ respectively, and a learning rate $\alpha$ that is equal to $5$. Using this particular setting, (\ref{Hedge_minimization}) becomes
\begin{align*}
H(x) = \frac{x \exp\{- 5 \cdot x\}}{x \exp\{- 5 \cdot x\} + (1-x) \exp\{- 5 \cdot 10 \cdot (1-x)\}}.
\end{align*}
If $x \in (0,1)$ is a point in a periodic orbit of period three, applying $H$ three times to $x$ gives back $x$, that is, $H(H(H(x))) = x$ or $H^3(x) = x$. Figure~\ref{saldkjfnalxdkhjf} shows the fixed points of $H^3$; these include the fixed point of $H$ (as a point of period one trivially also has period three) and, in addition, contain a periodic orbit of period three as indicated by the arrows shown in the figure. 

\begin{figure}[tb]
\centering
\includegraphics[width=8cm]{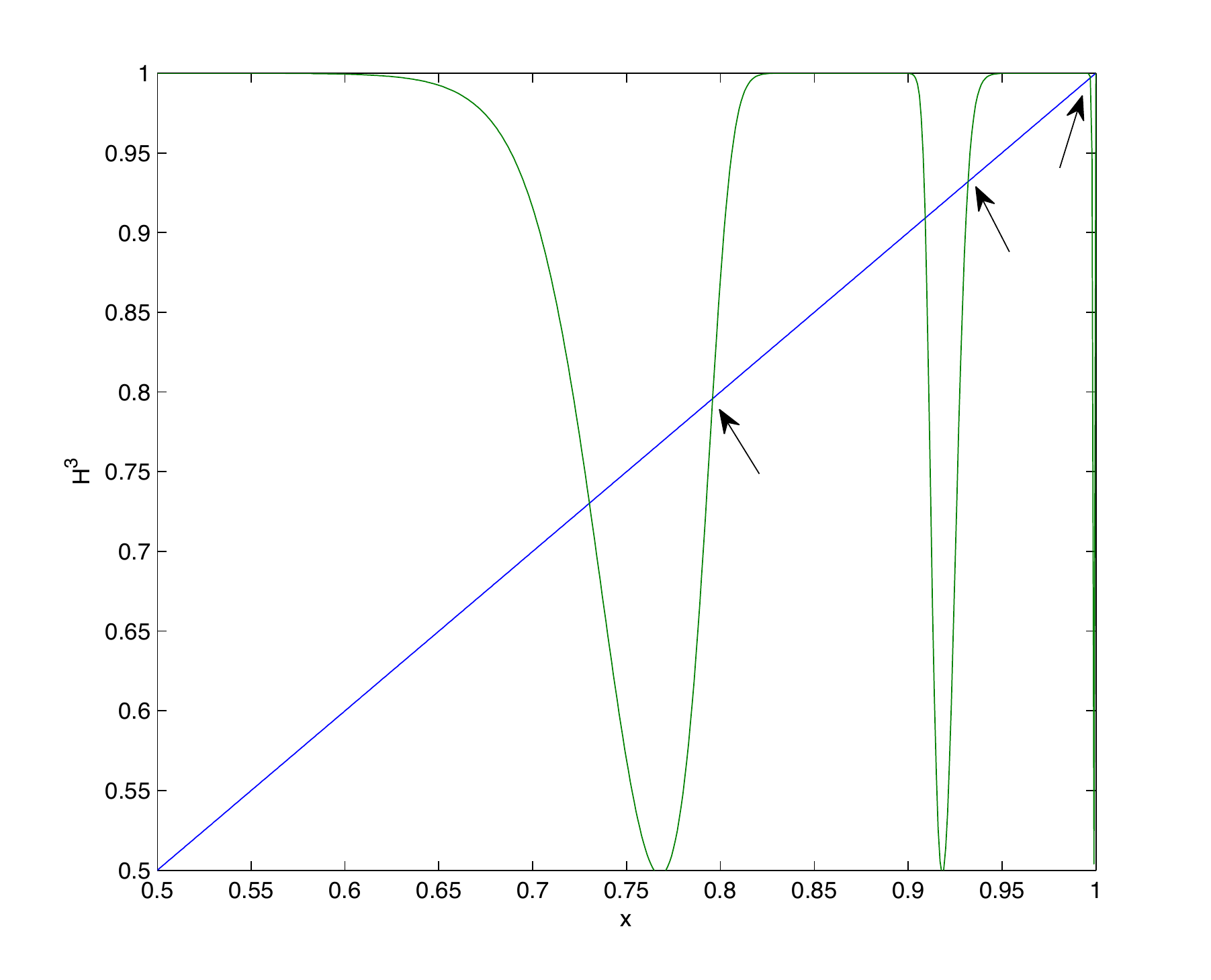}
\caption{\label{saldkjfnalxdkhjf}
$H^3(x)$ as a function of $x$. The arrows indicate a periodic orbit of period three.}
\end{figure}

\subsection{Linear stability in a simple selfish routing game}

In the remainder of this section, we consider a selfish routing game consisting of parallel links (arcs) between a pair of vertices, a source and a destination. We assume that the amount of flow from source to destination is equal to one (without loss of generality) and that the link cost functions are affine. Our goal is to show that the Wardrop equilibrium of such a selfish routing game is attractive under Hedge using an approach based on the linearization of the corresponding dynamical system.

\subsubsection{Summary of analytical approach}

The first step consists of computing the Jacobian matrix of the dynamical system that arises as we use Hedge to distribute the (unit) amount of flow across the system of parallel links; the Jacobian matrix is naturally computed at the corresponding Wardrop equilibrium (which is easily shown to be unique provided the slopes of the cost functions are strictly positive). Our objective is to show that the eigenvalues of such Jacobian matrix are (strictly) inside the unit disk. (We achieve this objective for all but a subset of systems of parallel links of affine cost functions that has measure zero.) In order to bound the area where the eigenvalues are located we rely on the {\em Perron-Frobenius theorem} and the {\em intermediate value theorem}. According to the Perron-Frobenius theorem, a nonnegative matrix (that is, a matrix with nonnegative entries) has a {\em real} eigenvalue that is equal to its {\em spectral radius} (defined as the modulus of the largest-in-modulus eigenvalue). It turns out that for a small enough learning rate, the Jacobian is indeed a nonnegative matrix, which implies its largest-in-modulus eigenvalue is real. The rest of the proof aims at showing that this eigenvalue is (strictly) smaller than one (which, therefore, implies that {\em all} eigenvalues have a modulus smaller than one and, thus, that the Wardrop equilibrium is asymptotically stable with a geometric rate of convergence near the equilibrium). This is accomplished through a continuity argument of the eigenvalues of the Jacobian and the intermediate value theorem.

\subsubsection{Analysis}
\label{linear_analysis}

Recall (from dynamical systems theory) that the Jacobian $J$ of a map $M$ whose domain and range are subsets of $\mathbb{R}^m$ (the $m$-dimensional Euclidean space) at a point $x$ of the domain is given by:
\[ 
J := 
\left[ \begin{array}{ccccc}
\frac{\partial M_1}{\partial x_1}  & \frac{\partial M_1}{\partial x_2} & \cdots & \frac{\partial M_1}{\partial x_m} \\
\frac{\partial M_2}{\partial x_1}  & \frac{\partial M_2}{\partial x_2} & \cdots & \frac{\partial M_2}{\partial x_m} \\
\vdots & \vdots & \vdots & \vdots \\
\frac{\partial M_m}{\partial x_1}  & \frac{\partial M_m}{\partial x_2} & \cdots & \frac{\partial M_m}{\partial x_m} \\
\end{array} \right] \] 
where the partial derivatives are evaluated at $x$. As noted earlier, since selfish routing games are more appropriately viewed from a minimization perspective, we taken on that perspective in this section using lower case letters, e.g., $x$, to denote flows rather than capital case. We will also use $j$ as the running index of a link (instead of $i$). Let us give the formula for the Jacobian at a Wardrop equilibrium (in a routing game consisting of parallel links) of full support (that is, an equilibrium such that every link carries population mass). Let $x = (x_1,\ldots,x_m)$ be such a full-support equilibrium (and, therefore, also a fixed point of Hedge). Assuming (for simplicity and without loss of generality) unit demand, the Jacobian $J_m$ of \eqref{Hedge_minimization} at $x$ is given by:
\begin{align}
J_m(i,j) = \left\{ \begin{array}{ll} 
                           \left( 1 - \alpha x_i \frac{\partial c_i(x_i)}{\partial x_i} \right) (1- x_i), &  i = j\\ 
                           - x_i \left( 1 - \alpha x_j \frac{\partial c_j(x_j)}{\partial x_j} \right), &  i \neq j.
                        \end{array}
                \right.\label{Jacobian-full-support}
\end{align}
Here $i$ denotes a row of the Jacobian matrix and $j$ a column. $J_m$ has the eigenvalue zero (it is easy to verify its rows sum to zero). We find it convenient to use {\em deflation} (see, for example, see \citep{Horn}) to obtain a square matrix $K_m$, of size
one smaller, whose eigenvalues are the remaining eigenvalues of $J_m$.
It is possible to perform the deflation in such a way that the resulting matrix
$K_m$ is given by the following expression:
\begin{align}
K_m = I - \alpha \hat{K}_m\label{zxcvmnsdfklj}
\end{align}
where
\[
\hat{K}_m(i,j) = \left\{ \begin{array}{ll} 
                           x_1 x_{i+1} \frac{\partial c_1}{\partial x_1} + x_{i+1} (1-x_{i+1}) \frac{\partial c_{i+1}}{\partial x_{i+1}} , &  i = j\\ 
                           x_{i+1} \left( x_1 \frac{\partial c_1}{\partial x_1} - x_{i+1} \frac{\partial c_{i+1}}{\partial x_{i+1}} \right), &  i \neq j.
                        \end{array}
                \right.\\
\]

In a system of $m$ parallel links whose cost functions are affine and strictly increasing (that is, $c_j(x_j) = \rho_j + \sigma_j x_j, \sigma_j > 0$), a full-support equilibrium is the unique solution of the following system of equations:
\[
\left[ \begin{array}{ccccc}
\sigma_1 & -\sigma_2 & 0 & \cdots & 0\\
\sigma_1 & 0 & -\sigma_3 & \cdots & 0\\
\vdots & \vdots & \vdots & \vdots & \vdots\\
\sigma_1 & 0 & 0 & \cdots & -\sigma_m\\
1 & 1 & 1 & \cdots & 1\\
\end{array} \right]
\left[ \begin{array}{cc}
x_1\\
x_2\\
\vdots\\
x_{m-1}\\
x_m\\
\end{array} \right]=
\left[ \begin{array}{cc}
\rho_2 - \rho_1\\
\rho_3 - \rho_1\\
\vdots\\
\rho_m - \rho_1\\
1\\
\end{array} \right].
\]
The solution for $j=1,\ldots,m$ is
\begin{align}
x_j = \frac{\sum_{i \neq j} (\rho_{i} - \rho_j) \prod_{k \neq i,j} \sigma_k + \prod_{i \neq j} \sigma_i}{\sum_{1 \leq i_1 < \cdots < i_{m-1} \leq m} \sigma_{i_1} \cdots \sigma_{i_{m-1}}}.\label{solution-of-linear-system}
\end{align}
Furthermore, $K_m$ simplifies as follows:
\begin{align}
K_m(i,j) = \left\{ \begin{array}{ll} 
                           1-\alpha x_{i+1}(\sigma_{i+1} + \rho_{i+1} - \rho_1), &  i = j\\ 
                           \alpha x_{i+1}(\rho_1 - \rho_{i+1}), &  i \neq j.
                        \end{array}
                \right.\label{asdklfjslkdfjd}
\end{align}

In the sequel, we show that the eigenvalues of $J_m$ are inside the unit disk for a sufficiently small learning rate $\alpha$. Let us, to that end, recall the definition of the {\em spectral radius}.

\begin{definition}
If $\lambda_i,  i = 1,\ldots,m$ are the eigenvalues of a square matrix A, then its spectral radius
$\rho(A)$ is defined as $\rho(A) := \max_i \{ |\lambda_i| \}$.
\end{definition}

Our goal is to prove the following theorem (considering the case of full-support equilibria first).

\begin{theorem}
\label{stability-theorem-linear-costs}
There exists an upper bound $\bar{\alpha} > 0$ on the learning rate such that, if the learning rate $\alpha$ is not greater than this upper bound (i.e., $0 < \alpha \leq \bar{\alpha}$), the spectral radius of the Jacobian at a Wardrop equilibrium of full support in a system of $m$ parallel links whose cost functions are affine (and have positive slopes) is smaller than one.
\end{theorem}

Once the proof of this theorem is complete, we will follow up with the case of Wardrop equilibria having partial support (that is, the equilibrium flow in at least one link is zero) rather easily.\\

Proceeding with the proof, we find first an the upper bound on the learning rate such that for all values of the learning rate up to this upper bound, the Jacobian is a nonnegative matrix. This implies, by the Perron-Frobenius theorem, that, for all such values of the learning rate, the Jacobian has a real and positive eigenvalue that is equal to its spectral radius. 

\begin{lemma}
\label{zmvbzmxcvb}
Consider a system of $m$ parallel links such that $c_j(x) = \rho_j + \sigma_j x_j, \sigma_j > 0, j=1\ldots,m$ whose Wardrop equilibrium has full support and assume the links have been numbered such that $\rho_1 \geq \rho_j, j=2,\ldots,m$. Then
there exists $\bar{\alpha}$ such that,
if $0 < \alpha \leq \bar{\alpha}$, $K_m$ has an eigenvalue
that is equal to its spectral radius. 
\end{lemma}

\begin{proof}
Since $\rho_1 \geq \rho_j, j=2,\ldots,m$, the off-diagonal entries of $K_m$ 
are nonnegative. Since 
$\rho_j + \sigma_j > \rho_j + \sigma_j x_j = \rho_1 + \sigma_1 x_1 > \rho_1$,
if
\begin{align} 
\bar{\alpha} = \min_{j \in \{2,\ldots,m\}} \left\{ \frac{1}{x_j(\sigma_j + \rho_j - \rho_1)} \right\},\label{alpha_upper_bound}
\end{align}
then if $0 < \alpha \leq \bar{\alpha}$ the diagonal entries are also nonnegative, and, therefore,
$K_m$ is nonnegative. Thus, by an application of the Perron-Frobenius
theorem (see, for example, \citep[p. 503]{Horn}), $K_m$ has a (real and positive) eigenvalue
that is equal to its spectral radius. 
\end{proof}

Then, to show that the spectral radius is smaller than one (and, therefore, all eigenvalues are inside the unit disk) we rely on the {\em intermediate value theorem}. We apply this theorem to a function that maps couples consisting of the vectors of offsets and slopes of the links to the largest in magnitude eigenvalue (that is real and positive) assuming the learning rate is equal to the aforementioned upper bound. Let's denote this function by $f$. Letting $\mbox{\boldmath${\rho}$} = (\rho_1,\ldots,\rho_m)$ and $\mbox{\boldmath${\sigma}$} = (\sigma_1,\ldots,\sigma_m)$, the domain of $f$ is the space of all parameters $(\mbox{\boldmath${\rho}$}, \mbox{\boldmath${\sigma}$})$ corresponding to systems of $m$ parallel links whose Wardrop equilibrium has full support. Furthermore, 
\begin{align*}
f(\mbox{\boldmath${\rho}$}, \mbox{\boldmath${\sigma}$}) = \rho(K_m(\bar{\alpha}(\mbox{\boldmath${\rho}$}, \mbox{\boldmath${\sigma}$}), \mbox{\boldmath${\rho}$}, \mbox{\boldmath${\sigma}$})).
\end{align*}

First, we show that this function is continuous, then we show that the domain of this function is connected, and finally we show that the Jacobian cannot have an eigenvalue equal to one. 

\begin{lemma}
\label{continuity-argument}
$f$ is a continuous function of $\mbox{\boldmath${\rho}$}$ and $\mbox{\boldmath${\sigma}$}$.
\end{lemma}

\begin{proof}
The eigenvalues of a matrix are continuous functions of the elements of the matrix, and the elements of $K_m(\bar{\alpha}(\mbox{\boldmath${\rho}$}, \mbox{\boldmath${\sigma}$}), \mbox{\boldmath${\rho}$}, \mbox{\boldmath${\sigma}$})$ are continuous functions of $\mbox{\boldmath${\rho}$}$ and $\mbox{\boldmath${\sigma}$}$. Therefore, by composition, the eigenvalues of $K_m$ are continuous functions of $\mbox{\boldmath${\rho}$}$ and $\mbox{\boldmath${\sigma}$}$.
\end{proof}

Let's denote the domain of $f$ by $S$. We have
\begin{align*}
S = \{(\mbox{\boldmath${\rho}$}, \mbox{\boldmath${\sigma}$}) \in \mathbb{R}^m \times \mathbb{R}^m_+ | x(\mbox{\boldmath${\rho}$}, \mbox{\boldmath${\sigma}$}) \mbox{ has full support } \}
\end{align*}
where $\mathbb{R}^m$ is the $m$-dimensional Euclidean space and $\mathbb{R}_+^m$ is its positive orthant.

\begin{lemma}
\label{woirwortwoirjtj}
$S$ is connected.
\end{lemma}

\begin{proof}
The set $\{ (\mbox{\boldmath${\rho}$}, \mbox{\boldmath${\sigma_0}$}) | x_j > 0, j=1,\ldots,m, \mbox{\boldmath${\sigma_0}$ is constant} \}$
that results if we fix the slopes and let the offsets vary
is the intersection of $m$ open half-spaces (see Equation~\ref{solution-of-linear-system}) and, therefore, it is convex, and, thus, connected. 
The set $\{ (\mbox{\boldmath$0$}, \mbox{\boldmath${\sigma}$}) | x_j > 0, j=1,\ldots m \}$ that results if we fix the
offsets to $0$ and let the slopes vary is also
connected since it is homeomorphic to the open positive orthant of $\mathbb{R}^m$.
Therefore, $S = \{ (\mbox{\boldmath${\rho}$}, \mbox{\boldmath${\sigma}$}) | x_j > 0,j=1,\ldots,m \}$ is  
connected since it is possible to connect any two points $(\mbox{\boldmath${\rho}$}, \mbox{\boldmath${\sigma}$})$ and $(\mbox{\boldmath${\rho'}$},\mbox{\boldmath${\sigma'}$})$ in $S$ as follows: $(\mbox{\boldmath${\rho}$},\mbox{\boldmath${\sigma}$}) \rightarrow (\mbox{\boldmath${0}$},\mbox{\boldmath${\sigma}$}) \rightarrow (\mbox{\boldmath${0}$},\mbox{\boldmath${\sigma'}$}) \rightarrow (\mbox{\boldmath${\rho'}$},\mbox{\boldmath${\sigma'}$})$.
\end{proof}

\begin{lemma}
\label{zxclovhanm2}
$K_m$ does not have the eigenvalue $1$ for any value of the learning rate.
\end{lemma}

\begin{proof}
To prove that $K_m$ does not have the eigenvalue $1$,
because of equation~(\ref{zxcvmnsdfklj}),
it suffices to prove that $\hat{K}_m$, given by following expression,
\begin{align*}
\hat{K}_m = 
\left[ \begin{array}{cccc}
 x_2(\sigma_2 + \rho_2 - \rho_1) & \cdots & x_2(\rho_2 - \rho_1) \\
 \vdots & \vdots & \vdots \\
 x_m(\rho_m - \rho_1) & \cdots & x_m(\sigma_m + \rho_m - \rho_1) \\
\end{array} \right],\label{breaking-tilde-Jm}
\end{align*} 
does not have the eigenvalue $0$. We can write $\hat{K}_m$
as the product of a diagonal matrix whose diagonal entries are $x_2, \ldots, x_m$
and the following matrix:
\begin{align*}
K_m' = 
\left[ \begin{array}{cccc}
 \sigma_2 + \rho_2 - \rho_1 & \cdots & \rho_2 - \rho_1 \\
 \vdots & \vdots & \vdots \\
 \rho_m - \rho_1 & \cdots & \sigma_m + \rho_m - \rho_1 \\
\end{array} \right].
\end{align*}
Because $x_j > 0, j=1,\ldots,m$ by assumption, it suffices to show that
$\det K_m' \neq 0$. We show that \[\det K_m' = (-1)^{m-1} \det N(x_1)\] where
\begin{align*}
N(x_1) = 
\left[ \begin{array}{cccccc}
\rho_2 - \rho_1 & -\sigma_2 & 0 & \cdots & 0\\
\rho_3 - \rho_1 & 0 & -\sigma_3 & \cdots & 0\\
\vdots & \vdots & \vdots & \vdots & \vdots\\
\rho_m - \rho_1 & 0 & 0 & \cdots & -\sigma_m\\
1 & 1 & 1 & \cdots & 1\\
\end{array} \right].
\end{align*}
It is easy to show that $\det N(x_1)$ is the numerator of $x_1$ in~(\ref{solution-of-linear-system}).
Therefore, showing that $\det K_m' = (-1)^{m-1} \det N(x_1)$
will suffice to complete the proof because $x_1 \neq 0$ by assumption.

We perform elementary column operations to $N(x_1)$, multiplying all columns except
for the first one by $-1$ and then adding the first column to all other columns, to obtain
the following matrix
\begin{align*}
\left[ \begin{array}{ccccc}
\rho_2 - \rho_1 & \sigma_2 + \rho_2 - \rho_1  & \cdots & \rho_2 - \rho_1\\
\rho_3 - \rho_1 & \rho_3 - \rho_1 & \cdots & \rho_3 - \rho_1\\
\vdots & \vdots & \vdots & \vdots\\
\rho_m - \rho_1 & \rho_m - \rho_1 & \cdots & \sigma_m + \rho_m - \rho_1\\
1 & 0 & \cdots & 0\\
\end{array} \right]
\end{align*}
whose determinant is equal to $\det K_m'$. This completes the proof.
\end{proof}

We are now ready to conclude, by the intermediate value theorem, that either all systems of parallel links with affine cost functions and positive slopes are unstable (have a spectral radius greater than one) or they are asymptotically stable. We prove the latter through giving an example of a system that is asymptotically stable. The example is given in the next lemma.

\begin{lemma}
\label{stability-of-nash-equilibria-linear-costs}
In a system consisting of $m$ parallel links such that $c_j(x_j) = \sigma_j x_j, \sigma_j > 0, j = 1,\ldots,m$, there exists an $\bar{\alpha} > 0$ such that, if $0 < \alpha \leq \bar{\alpha}$, the eigenvalues of $K_m$ are inside the unit disk.
\end{lemma}

\begin{proof}
First note that the Wardrop equilibrium has full support (in all such systems of parallel links), and, therefore, $K_m$ is given by ~(\ref{asdklfjslkdfjd}). Now observe that if $\rho_1 = \cdots = \rho_m = 0$, then $K_m$ is diagonal, and, therefore, if $\alpha < \beta \min\{ \frac{2}{x_j \sigma_j} \}, j=2,\ldots,m$, for any $\beta$ such that $0 < \beta < 1$, then $K_m$ has spectral radius smaller than $1$. This completes the proof.
\end{proof}

Observe that the proof of Theorem~\ref{stability-theorem-linear-costs} is now complete.\\

Finally, we consider the Jacobian at a fixed point that has partial support. (Recall that the fixed points of Hedge do not necessarily coincide with Wardrop equilibria, but they are related as noted earlier.)
At a fixed point with partial support, let's number the links such that $I_+ = \{ 1, \ldots, p \}$ is the set of indices
of links with positive flow and $I_0 = \{ p+1, \ldots, m \}$ is the set of indices of links
with zero flow. Then we can write
\begin{align*}
J_m = \left[ \begin{array}{ccc}
 J_m(I_+, I_+)  & J_m(I_+, I_0)\\
 J_m(I_0, I_+)  & J_m(I_0, I_0)\\
\end{array} \right] =
\left[ \begin{array}{ccc}
 J_m^{++}  & J_m^{+0}\\
 J_m^{0+}  & J_m^{00}\\
\end{array} \right].
\end{align*}
Using calculus we find that $J_m^{0+}$ is the zero matrix, and, therefore, $J_m$ is blog diagonal
\begin{align*}
J_m = \left[ \begin{array}{ccc}
 J_m^{++}  & J_m^{+0}\\
 0  & J_m^{00}\\
\end{array} \right] 
\end{align*}
where
\begin{align*}
J_m^{++}(i,j) &= \left\{ \begin{array}{ll} 
                           \left( 1 - \alpha x_i \frac{\partial c_i(x_i)}{\partial x_i} \right) (1- x_i), &  i = j\\ 
                           - x_i \left( 1 - \alpha x_j \frac{\partial c_j(x_j)}{\partial x_j} \right), &  i \neq j
                        \end{array}
                \right.
\end{align*}
so that $J_m^{++}$ is the Jacobian of a full-support equilibrium in the system consisting of the first $p$ links and
\begin{align*}
J_m^{00}(i,j) &= \left\{ \begin{array}{ll} 
                           \exp\{ - \alpha (c_{i}(x_{i}) - c_1(x_1)), &  i = j\\ 
                           0, &  i \neq j
                        \end{array}
                \right.
\end{align*}
so that $J_m^{00}$ is a diagonal matrix whose diagonal entries are
$\exp\{ - \alpha (c_{j}(0) - c_1(x_1)\}, j = p+1,\ldots,m$.

Because $J_m$ is block diagonal, its eigenvalues are equal to the eigenvalues of $J_m^{++}$ and
the eigenvalues of $J_m^{00}$, and, therefore, we can say the following:

\begin{itemize}

\item If there exists $j\in\{p+1,\ldots, m\}$ such that $c_j(0) < c_1(x_1)$, then,
the fixed point (which is not a Wardrop equilibrium) is unstable (the spectral radius of the Jacobian is larger than one). 

\item If, for all $j\in\{p+1,\ldots, m\}$, $c_j(0) > c_1(x_1)$, then the stability of the
fixed point is determined by the stability of $J_m^{++}$. Therefore,
by Theorem~\ref{stability-theorem-linear-costs}, the fixed point is asymptotically stable (the eigenvalues of the Jacobian are inside the unit disk).

\item If $c_j(0) \geq c_1(x_1)$ for all $j\in\{p+1,\ldots, m\}$ and there exists a
$j\in\{p+1,\ldots, m\}$ such that $c_j(0) = c_1(x_1)$, then
$J_m$ has an eigenvalue equal to one in which case the Jacobian test is inconclusive. 
However, note that this case is rare, as adding an arbitrarily small positive or negative offset to the cost function of link $j$ either changes the equilibrium to a full-support one or brings us to one of the previous two cases.
\end{itemize}

In closing this section, we note that although the previous analysis is important in the sense of not only establishing asymptotic stability but also a geometric rate of convergence near the equilibrium, this type of analysis is evidently quite difficult to generalize to the (quite general) cases that our (much less complicated) argument based on Lyapunov functions was able to handle.

\subsection{An intuitive interpretation of evolutionary stability}
\label{intuition}

Evolutionary stability originated as a branch of mathematical biology that is typically motivated in the literature in the setting of evolution wherein payoffs correspond to the fitness of organisms are measured by the number of offspring they can bear. Organisms, drawn from an infinite population thereof, are thought to interact in pairs (corresponding to the players positions in a symmetric bimatrix game) and the main assumption is that such interactions affect the organisms' ability to reproduce and raise offspring. How this approach gives rise to the notion of evolutionary stability is explained in sufficient detail in the introduction of this paper. However, as also noted in the introduction, the approach has found far reaching applications in economics and social evolution. 

In the rest of this section, we are concerned with an interpretation of evolutionary stability oriented toward a practical problem in networked systems design, namely, how to route traffic between sources and destination in a computer network such as the Internet. It is perhaps meaningful to recall at this point that algorithmic game theory was originally motivated by a quote one of its founding fathers, namely, Christos Papadimitriou attributes to one of the Internet's architects, namely, Scott Shenker: ``The Internet is an equilibrium, we just have to identify the game.''

Wardrop equilibria (that is, Nash flows in nonatomic congestion games), have been studied as a possible model of the traffic pattern in communication networks, such as the Internet, serving selfish users. Perhaps the most prominent example where the selfishness of Internet users has manifested itself is in Internet routing; the Internet uses protocols to route user traffic based on metrics that are not necessarily aligned with the incentives of the users, and users have responded by, for example, deploying overlay networks to override these metrics. This phenomenon in Internet architecture is the motivation for the inquiry in the remainder of this section.

We should note that the main result we prove in this section, albeit original, admits a shorter alternative derivation using notions of equivalence between characterizations of evolution stability in the general setting of population games. Our main contribution, therefore, primarily lies in the discussion and intuition of the analytical approach used in the derivation.

\subsubsection{Routing and architectural stability}

One of the goals of algorithmic game theory is to provide an understanding of the workings and properties of the Internet \citep{Algorithmic-Game-Theory}. {\em Routing}, that is, determining the paths along which traffic is sent from sources to destinations, is perhaps the single most important function that supports a network. Routing has been studied from a game-theoretic perspective in the context of {\em congestion games} wherein a population of agents select paths in a network of congestible resources (in particular, communication links) so as to selfishly minimize delay.

The congestion games of interest in this section are {\em nonatomic}. In a nonatomic congestion game, the population of agents forms a continuum so that each agent controls an infinitesimal amount of load. In these games, as mentioned earlier, Nash equilibria are also called {\em Wardrop equilibria}, due to the early studies of Wardrop in road traffic transportation systems~\citep{Wardrop}. Under mild assumptions, Wardrop equilibria are the unique global minimizers of a strictly convex function.

One of our goals in studying nonatomic congestion games is to provide a game-theoretic view of {\em architectural properties} of routing protocols. For our purposes, a routing protocol is a collection of probability distributions over the paths of a network (where each probability distribution corresponds to a source-destination pair) that is used to determine the routes packets traverse from source to destination. Packets form a population of mass equal to the demand. Then, intuitively by the law of large numbers, a routing protocol induces a {\em flow}. 

In modeling a routing protocol in this way we assume that flow demand is {\em splittable} that, in practice, may interfere with the current implementation of enduser processes running on top of the routing protocol (such as, for example, congestion control as that is implemented in TCP). However, there is an ongoing effort for a possibly major overhaul of the Internet architecture (for example, see \citep{Clean-Slate} for a discussion), and splittable flows have important architectural advantages. Therefore, we believe we are making a legitimate assumption.

In inducing a flow, a routing protocol also induces a {\em social cost} that an optimal routing protocol would minimize. However, Internet traffic is generated and consumed by endusers that are independent self-interested players willing to deviate from the decisions of the routing protocol if they see fit. For example, a subset of the endusers may decide to circumvent the incumbent routing protocol's decisions through deploying an {\em overlay routing protocol} to get better performance. 

In this environment, game theory predicts that optimal routing is {\em architecturally unstable} (in the sense that there exist users whose performance would improve by deviating), and that instead enduser traffic is going to be routed according to a Nash equilibrium. However, game theory's prediction of equilibrium play leaves many important questions open.

In this section, we study Wardrop equilibria from an {\em evolutionary perspective}. One of the main concerns of evolutionary game theory is the study of how {\em invasions} (also called {\em mutations}) fare against an incumbent population when the invaders and incumbents interact. This is a particularly appropriate perspective to reason about questions regarding the {\em architectural stability} of routing protocols. The term `architectural stability' is to be understood as {\em stability to invasions}, where an invasion is an {\em architectural deviation}, that is, a change in the probability distribution that determines the routes of a subset of the packets. This abstraction is independent of  the implementation of any particular network architecture, enabling us to make precise mathematical statements. 

Using such a mathematical formulation, we are interested in answering questions such as: `Are Wardrop equilibria architecturally stable?' or `Can Wardrop equilibria be the natural outcome of architectural evolution?' In this section, we attempt a formal inquiry on the {\em architectural stability} of network flows. Informally, a flow is architecturally stable if it is induced by a routing protocol whose decisions users would not have an incentive to override. Our formal definition of architectural stability is based on that of {\em evolutionary stability}, i.e., stability to invasions by {\em mutations} ({\em architectural deviations}), drawing on the aforementioned concepts from evolutionary game theory. 

Referring back to the classical interpretation of evolutionary stability based on competition between organisms and consider an analogy in which {\em organisms} correspond to {\em packets}. The strategies {\em incumbent packets} follow correspond to the routing decisions of the incumbent routing protocol and these decisions induce a flow, say, $x$. Suppose now that a {\em mutation (architectural deviation)} emerges in which a fraction $\epsilon$ of the packets {\em switch} to a different strategy, and suppose that the induced flow due to this deviation is $\epsilon y$. Then the total flow in this (bimorphic) network consisting of two routing protocols (the incumbent and the mutant) is $z_{\epsilon} = (1-\epsilon) x + \epsilon y$. 

Payoffs are (the negative of) {\em congestion costs} (delays). Then the natural criterion to determine whether the architectural deviation will expand or decline is the relative delay of the mutant packets, say $y \cdot c(z_\epsilon)$, with respect to the delay, say $x \cdot c(z_{\epsilon})$, of the incumbents. That is, if $y \cdot c(z_\epsilon) < x \cdot c(z_{\epsilon})$, the mutation will expand, and if $y \cdot c(z_\epsilon) > x \cdot c(z_{\epsilon})$, it will decline. 

It is then natural to define {\em architecturally stable flows} in a manner analogous to the previous definition of {\em evolutionary stable strategies}, taking into account that packets do not play a normal form game (as in the standard setting wherein, as mentioned earlier, evolutionary stability is discussed) but instead participate in a congestion game. Therefore, an architecturally stable flow is, informally, one that is stable to architectural deviations (mutations). This definition gives precise meaning to our previous questions regarding the architectural properties of Wardrop equilibria.

We answer these questions with the (perhaps surprising) result that Wardrop equilibria are {\em evolutionary dominant} in the following sense: Suppose that the routing decisions in a network are made according to two routing protocols, one routing according to the Wardrop equilibrium and the second routing according to an arbitrary different strategy. We show that the expected cost of a packet routed according to the first routing protocol is lower than the expected cost of a packet routed according to the second routing protocol irrespective of the second protocol's routing strategy and irrespective of the fraction of packets that it routes. 

That Wardrop equilibria are dominant has two important practical implications. The first is that Wardrop routing is architecturally stable as no architectural deviation can benefit when the Wardrop flow is the incumbent. The second is that architectural deviations based on Wardrop routing can be successfully {\em incrementally deployed} against any other incumbent flow in the sense that incumbent flows that are not at a Wardrop equilibrium are vulnerable themselves to invasion by Wardrop routing. Therefore, to state that Wardrop equilibria are the only natural outcome of user-driven architectural evolution would perhaps not be an overstatement.

\subsubsection{Evolutionary dominance of Wardrop equilibria}

Let us now prove the aforementioned claim. We start by defining the notions of routing strategies and induced flows. Then we gradually define `evolutionary dominance', starting with the concepts of `invasion' and `evolutionary stability'. The link between {\em evolutionary} and the notion of {\em architectural stability} as we earlier defined it is created by the observation that an evolutionary stable flow resists incremental architectural changes (invasions). Finally, we show a flow is dominant if and only if it cannot be invaded, and use this as the basis of the proof of the main result.

We may view the aggregate demand $\sum_i \omega_i$ as being allocated to the paths using {\em routing strategies}, which are probabilistic rules that determine the paths packets follow: The routing strategy of commodity $i$ is $(1/\omega_i) X_i \equiv \bar{X}_i$, which is a probability vector. A routing strategy for the network is a collection of probability distributions, one for each commodity. We will denote a network routing strategy by $\bar{X}$ whenever it corresponds to multi-commodity flow $X$.

We may think of a {\em routing protocol} as an arbitrary collection of routing strategies, one for each commodity. A routing protocol induces a flow, and, therefore, it is equivalent to think of a routing protocol as a routing strategy. From now on we will not use the term `routing protocol', and we will use instead the terms `routing strategy' and `flow'. We are going to consider {\em bimorphic populations} whose packets use one of two routing strategies, say $\bar{X}$ and $\bar{Y}$. It is then straightforward to show that any convex combination $(1-\epsilon) \bar{X} + \epsilon \bar{Y}$ is also a routing strategy. We may think of the packets of a bimorphic population as being partitioned into two subsets, the packets of flow $X$ (using routing strategy $\bar{X}$) and the packets of flow $Y$ (using routing strategy $\bar{Y}$). We will say that the former packets are ``routed according to $X$'' and that the latter packets are ``routed according to $Y$.''

We make the following notational conventions: First, $c(x|y) \doteq x \cdot c(y)$, where $x, y \in \mathbb{X}$, and, second, $c(\bar{x}|y) \doteq \bar{x} \cdot c(y)$, where $\bar{x}$ is the probability vector corresponding to $x \in \mathbb{X}$ and $y \in \mathbb{X}$. Using this notation, it is straightforward to show that $c(x|z) \leq c(y|z)$ if and only if $c(\bar{x}|z) \leq c(\bar{y}|z)$. With this background and intuition in mind, let us proceed to the definition of evolutionary dominance.

\begin{definition}
\label{xcnbvlzfjv}
We say that $y \neq x$ invades $x$ if $c(\bar{y}|x) < c(\bar{x}|x)$.
\end{definition}

Intuitively, $y$ invades $x$ if the delay of a packet in the population of $x$ that switches its routing strategy from $\bar{x}$ to $\bar{y}$ is lower than the delay of all other packets in $x$. 

\begin{definition}
We say that $y \neq x$ has a positive invasion barrier against $x$, denoted $b(y|x)$, if there exists a $b(y|x) > 0$ such that for all $0 < \epsilon \leq b(y|x)$, $$c(\bar{x} | (1-\epsilon) x + \epsilon y) < c(\bar{y} | (1-\epsilon) x + \epsilon y).$$
\end{definition}

Think of $x$ as the {\em incumbent} flow and $y$ as the {\em mutant}. If $y$ has a positive invasion barrier against $x$, then clearly $y$ cannot invade $x$. (Use $\epsilon = 0$ and Definition~\ref{xcnbvlzfjv}.) However, a positive invasion barrier further implies that even if an $\epsilon$-fraction of the packets in $x$ (where $\epsilon \leq b(y|x)$) switch to routing strategy $\bar{y}$, then the incumbent packets have a lower delay than that of the mutants.

\begin{definition}
We say that $x$ is evolutionarily stable if there exists $\beta > 0$ such that all other flows $y$ have an invasion barrier $b(y|x) \geq \beta$ against $x$.
\end{definition}

In fact, this is also our definition of an {\em architecturally stable flow,} which, intuitively, is a flow resisting incremental architectural changes. This intuition is made precise in the following definition.

\begin{definition}
We say that $x$ is incrementally deployable against $y$ if, for all $\epsilon \in [0,1]$, $$c(\bar{x} | (1-\epsilon) y + \epsilon x) \leq c(\bar{y} | (1-\epsilon) y + \epsilon x).$$
\end{definition}

Intuitively, a flow $x$ is incrementally deployable\footnote{I should note that ever since these definitions relating the notions of evolutionary stability to network architecture where captured in 2011, the appropriate definition of incremental deployability, a term that originated in the networking literature, has troubled me a lot since then. I refer the reader interested in exploring different rigorous manifestations of the intuition behind this definition to \url{http://arxiv.org/abs/1601.03162}.} against an incumbent flow $y$, if as the incumbent packets gradually change their routing strategy from $\bar{y}$ to $\bar{x}$, the delay of the invading packets is not more than the delay of the incumbent packets.  Going back to our earlier statement that an evolutionarily stable flow resists incremental architectural changes, we can now state more precisely that if a flow is evolutionarily stable then no mutant flow is incrementally deployable against it. In the following definition, we introduce the concept of {\em dominance}.

\begin{definition}
We say that $x$ dominates $y$, if $b(y|x) = 1$. In this case, we say that $y$ has an infinite invasion barrier against $x$.
\end{definition}

Consider a population of packets in which some fraction of the packets are routed according to $x$ and the remaining packets are routed according to $y$. $x$ dominates $y$ if the delay of the former packets is lower than the delay of the latter packets irrespective of the population mix. 

\begin{definition}
We say that a flow is dominant if it dominates every other flow.
\end{definition}

Note that evolutionarily dominant flows are not only evolutionarily stable (that is, robust to invasions), but they are also incrementally deployable against all other incumbents in the sense that these incumbents are vulnerable themselves to invasion by the evolutionarily dominant flow. 

Let us now proceed to proving the evolutionary dominance of a Wardrop flow. The essential property of nonatomic congestion games giving rise to the existence of a dominant flow is captured in Lemma~\ref{sdfsdflksdjf}. In the lemma's statement we need the following definition.

\begin{definition}
Given $x$ and $y$ and $\epsilon \in [0,1]$ let $$\delta(\epsilon | x, y) = c(x | (1-\epsilon) x + \epsilon y) - c(y | (1-\epsilon) x + \epsilon y).$$
\end{definition}

\begin{lemma}
\label{sdfsdflksdjf}
$\delta(\epsilon| x, y)$ is a strictly decreasing function of $\epsilon$.
\end{lemma}

In the proof of Lemma~\ref{sdfsdflksdjf} we are going to use the following simple result whose (elementary) proof can, for example, be found in \citep{Roughgarden}.

\begin{lemma}
\label{bvxjzfhgaosfv}
$c(x|y) = \sum_{e \in E} c_e(y_e) x_e$.
\end{lemma}

\begin{proof}[Proof of Lemma~\ref{sdfsdflksdjf}]
By Lemma~\ref{bvxjzfhgaosfv}, we have
\begin{align*}
\delta(\epsilon| x, y) &= \sum_{e \in E} c_e((1-\epsilon) x_e + \epsilon y_e) x_e - \sum_{e \in E} c_e((1-\epsilon) x_e + \epsilon y_e) y_e\\
     &= \sum_{e \in E} c_e(x_e + \epsilon (y_e - x_e)) (x_e - y_e)
\end{align*}
and, therefore,
\begin{align*}
\delta'(\epsilon| x, y) = - \sum_{e \in E} (x_e - y_e)^2 c'_e(x_e + \epsilon (y_e - x_e)).
\end{align*}
Thus, by the assumption that the resource costs are strictly increasing, $\delta'(\epsilon| x, y) < 0$.
\end{proof}

Consider an incumbent flow $x$ and a mutant flow $y$. Lemma~\ref{sdfsdflksdjf} implies that, unless the mutants can invade when their mass is small, they will not be able to invade as their mass increases. In fact, Lemma~\ref{sdfsdflksdjf} also implies that even if the mutants can invade when their mass is small, as their mass increases, the incumbents might be able to get even or win. This is illustrated in Figure~\ref{saldkjfnalsdkhjf}, which shows an example of $\delta(\epsilon | x, y)$ assuming that the resource cost functions are linear. 

\begin{figure}[tb]
\centering
\includegraphics[width=12cm]{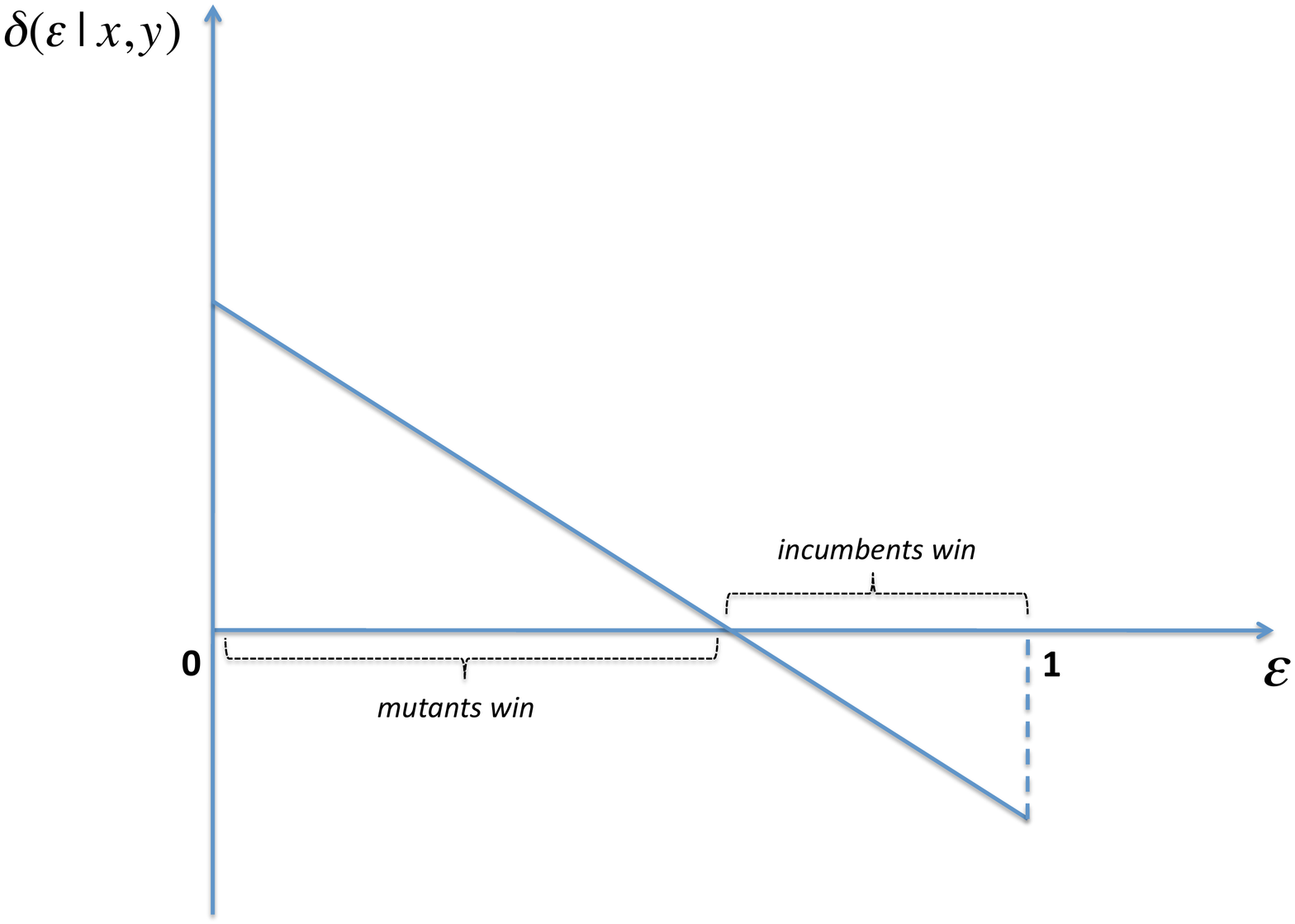}
\caption{\label{saldkjfnalsdkhjf}
$\delta(\epsilon|x,y)$ as a function of $\epsilon$.}
\end{figure}

\begin{lemma}
\label{zbkvjvsaio}
For any flows $x$ and $y \neq x$, either $y$ invades $x$ or $x$ dominates $y$.
\end{lemma}

\begin{proof}
The proof is clear from Figure~\ref{saldkjfnalsdkhjf}; if $y$ does not invade $x$, then $\delta(0 | x,y)$ is negative and $\delta(\epsilon | x,y)$ will always remain so. Formally, suppose $y$ does not invade $x$. Then $\delta(0 | x,y) = c(x|x) - c(y|x) \leq 0$, and, therefore, Lemma~\ref{sdfsdflksdjf} implies that, for all $\epsilon \in (0,1]$, $c(x | (1 - \epsilon) x + \epsilon y) < c(y | (1 - \epsilon) x + \epsilon y)$. Therefore, $y$ has an infinite invasion barrier against $x$, and, thus, $x$ dominates $y$.
\end{proof}

\begin{lemma}
\label{zxkcvboldkfn}
A flow is dominant if and only if it cannot be invaded.
\end{lemma}

\begin{proof}
If a flow can be invaded, then it is not evolutionary stable, and, therefore, it is not dominant. Suppose now that a flow cannot be invaded. There are two ways to show that it is dominant. The first is to observe that flow $x$ cannot be invaded if and only if $\delta(0 | x, y) \leq 0$ for all $y$. Dominance is then implied by Lemma~\ref{sdfsdflksdjf}. The second is to use Lemma~\ref{zbkvjvsaio} in a straightforward way.
\end{proof}

We are now ready to prove the dominance of Wardrop equilibria.

\begin{lemma}
\label{dfjglnsjfngldsk}
If a flow is evolutionary stable, then it is dominant.
\end{lemma}

\begin{proof}
Consider an evolutionarily stable flow $x$ and another flow $y$. By the definition of evolutionarily stability, $y$ has a positive invasion barrier against $x$, and, therefore, it does not invade $x$. Then, by Lemma~\ref{zbkvjvsaio}, $x$ dominates $y$. Since $y$ is arbitrary, the lemma is proved.
\end{proof}

\begin{theorem}
\label{pqowbfdnm}
The Wardrop flow is dominant.
\end{theorem}

\begin{proof}
\cite{SelfishRoutingEvolution} and \cite{StableGames} show that Wardrop flows are evolutionarily stable. The theorem is then implied by Lemma \ref{dfjglnsjfngldsk}.
\end{proof}

Let us finally relate the notion of evolutionary dominance with that of incremental deployability. Note, in this vein, that whereas evolutionary stability secures a ``defensive position'' (prevents invasion by mutants), evolutionary dominance further implies that ``aggression'' (against incumbents) is successful. Intuitively, a flow is {\em deployable} in a selfish environment if it is not unsuccessful (i.e., its packets do not fare worse than the incumbent packets). However, the structure of nonatomic congestion games implies that if a flow is deployable, then it must necessarily be successful against the incumbents. This intuition is captured in the following theorem.

\begin{theorem}
\label{asdohfoid}
Let $x \neq y$. $x$ dominates $y$ if and only if $x$ is incrementally deployable against $y$.
\end{theorem}

\begin{proof}
If $x$ dominates $y$, then, clearly, $x$ is also incrementally deployable against $y$. Suppose now that $x$ is incrementally deployable against $y$. Then, for all $\epsilon \in [0,1]$,
\begin{align*}
c(\bar{x} | (1-\epsilon) y + \epsilon x) \leq c(\bar{y} | (1-\epsilon) y + \epsilon x).
\end{align*}
This is equivalent to saying that, for all $\epsilon \in [0,1]$,
\begin{align*}
c(\bar{x} | (1-\epsilon) x + \epsilon y) \leq c(\bar{y} | (1-\epsilon) x + \epsilon y),
\end{align*}
or that, for all $\epsilon \in [0,1]$, $\delta(\epsilon|x,y) \leq 0$. However, by Lemma~\ref{sdfsdflksdjf}, $\delta(\epsilon|x,y)$ is a strictly decreasing function of $\epsilon$, and, therefore, for all $\epsilon \in (0,1]$, $\delta(\epsilon|x,y) < 0$, which completes the proof.
\end{proof}

Therefore, the previous theorem implies that, in nonatomic congestion games, dominance and incremental deployability (as we earlier defined it) are, in fact, equivalent concepts.

\subsubsection{Discussion}

Let us attempt to interpret the meaning that Wardrop equilibria are evolutionarily dominant in a networking environment. Suppose in this vein that routing in a network is such that the induced flow corresponds to some particular vector $x \neq x^*$, where $x^*$ is the Wardrop equilibrium. Suppose further that a network architect computes $x^*$ and distributes this value to the population of packets. Will each packet be willing to immediately switch from routing strategy $\bar{x}$ to $\bar{x}^*$? 

In the single-population case, the answer is a definitive yes (provided packets are greedy) as the delay of each packet performing such a switch improves. However, in the multi-population setting, each packet in the aggregate population thereof corresponds to a particular commodity. Assuming packets are independent selfish agents, a more appropriate question would be to ask whether each packet would be willing to switch from $\bar{x}_i$ to $\bar{x}_i^*$ where $i$ is the respective commodity. 

Pondering this question it becomes clear that our definition of evolutionary dominance does not ensure such a property necessarily holds for all commodities in the network but rather makes a statement that the average delay across all commodities will be superior as packets start switching. An interesting question that this observation raises is as follows: Does there exist a transition strategy from an arbitrary flow to the Wardrop flow during which every packet greedily benefits during the transition process? We believe there is, but we have not verified the conjecture.

\subsubsection{Evolutionary dominance and the Minty variational inequality}

Raising again the level of abstraction from selfish routing games to general population games, the notion of evolutionary dominance is of particular interest to multiplicative weights dynamics, as our results imply that evolutionarily dominant states are (interior) globally attractive under multiplicative weights dynamics. The notion of evolutionary dominance is also keenly related to the {\em Minty variational inequality,} which can be shown to be a special case of the standard definition of a variational inequality (which is sometimes referred to as a Stampacchia variational inequality).

\begin{definition}
Given a nonempty, closed, and convex set $\mathbb{X} \subseteq \mathbb{R}^\ell$ and a continuous vector field $F: \mathbb{R}^\ell \rightarrow \mathbb{R}^\ell$, the finite-dimensional {\em variational inequality problem} is to find an element $X^* \in \mathbb{X}$ such that
\begin{align*}
(X^* - X) \cdot F(X^*) \geq 0, \forall X \in \mathbb{X}.
\end{align*}
We call such an $X_*$ a {\em critical point} of $F$.
\end{definition}

\begin{definition}
Given a nonempty, closed, and convex set $\mathbb{X} \subseteq \mathbb{R}^\ell$ and a continuous vector field $F: \mathbb{R}^\ell \rightarrow \mathbb{R}^\ell$, the finite-dimensional {\em Minty variational inequality problem} is to find an element $X^* \in \mathbb{X}$ such that
\begin{align}
(X^* - X) \cdot F(X) \geq 0, \forall X \in \mathbb{X}.\label{MVIP}
\end{align}
We call such an $X_*$ a {\em Minty solution} of $F$.
\end{definition}

Minty solutions are necessary critical points in that they necessarily satisfy the standard definition of a variational inequality above (slightly extending our aforementioned definition of Nash equilibria in population games). Considering the setting of population games, a Minty solution is known as a {\em neutrally stable state}. If the inequality in \eqref{MVIP} is strict (and $\mathbb{X}$ is a simplotope), we obtain the definition of a GESS. Of course Minty solutions are not guaranteed to exist, as are global evolutionarily stable states. It can be shown, but the proof is omitted for brevity, that a GESS is evolutionarily dominant, using an appropriate extension of the definition of evolutionary dominance from the setting of selfish routing games. Noting that our asymptotic stability results in general population games imply that a GESS is (interior) globally asymptotically stable under the multiplicative weights dynamics, it is interesting to ask which classes of population games are equipped with evolutionarily dominant states. For example, potential games with a convex potential function are equipped with such a state, and it can be shown by an argument analogous to that used in Section \ref{Weissings_result} that population games with an interior ESS are also in this class. However, this list of examples is certainly not exhaustive; the question deserves further scrutiny.

We note that work along answering the previous question has already been taken up in the literature on variational inequalities \citep{Konnov, DanHad, John1, John2, Crespi}. However, before concluding this paper, we also note that a related (algorithmic in nature) question concerns recognizing population games equipped with an evolutionarily dominant state as well as recognizing games equipped with at least one evolutionarily stable state (which may not be dominant). In this vein, we point out that answering the latter question is thought to be a computationally hard problem: \cite{Conitzer} shows that even in symmetric bimatrix games this decision problem is $\Sigma_2^P$-complete, where $\Sigma_2^P$ is the complexity class at the second level of the computational hierarchy. To the extent of our knowledge, the former question is open. Furthermore, despite the aforementioned hardness result, the question of identifying special classes of games that can be shown to be always equipped with an ESS is certainly worth pursuing.

\section{Concluding remarks}
\label{Conclusion}

Motivated by the success of online learning theory in a variety of applications, in this paper, we set out to evaluate how one of the most well-known online learning algorithms, namely, Hedge (also sometimes called $MW$, where $MW$ stands for {\em multiplicative weights}), performs in the quite general setting of population games of which a variety of problems in game theory and nonlinear optimization theory are special cases. In this vein, we demonstrated convergence to evolutionarily stable states in the general setting, and went deeper into the analysis of more particular problems.

Our results admit, from an evolutionary perspective, an interpretation that a theory of evolution based on multiplicative weights is, in a sense, {\em backward compatible} with a theory of evolution based on the notion of evolutionary stability in that, in principle, the evolutionary phenomena that can be explained by the notion of evolutionary stability is a subset of the phenomena that evolutionary theories based on the notion of multiplicative update rules can explain. But our main result leaves open whether the theories meet in more particular evolutionary settings and to what degree. Answering this latter question is the subject of our ongoing work. 

Along these lines, we point out that the question has already been taken up from the perspective of continuous-time replicator dynamics. In particular, \cite{Zeeman} shows that the notions of asymptotic stability under the replicator dynamic and that of evolutionary stability do not coincide by giving an example of a Nash equilibrium in a symmetric bimatrix game that is not evolutionarily stable but is, however, asymptotic stability under the replicator dynamic. However, asymptotic stability under the replicator dynamic and evolutionary stability are known to be equivalent concepts in doubly symmetric bimatrix games \citep{Hofbauer-Sigmund}. It is, therefore, interesting to ask whether these results extend to the setting of multiplicative weights.

Our results, in particular, with respect to selfish routing games, suggest that multiplicative weights may well serve as a primitive in the interest of load balancing Internet traffic to reduce communication delay between sources and destinations. The evolution rule we analyze corresponds to the distributed setting of {\em synchronous} updates. As part of our ongoing work we are investigating, in the spirit of models presented by \cite{Bertsekas-Tsitsiklis} whether and to what degree such a synchrony assumption can be relaxed in our convergence and stability analysis.

Finally, multiplicative weights have found applications in the solution of linear and semidefinite programs \citep{FreundSchapire2, Kale}. Linear programming is equivalent to the computation of minimax equilibria in zero sum games, which are are not evolutionarily stable strategies. It is an interesting question whether our analysis can be extended to capture convergence to weaker notions of evolutionary stability, such as, for example, {\em neutral stability,} that would bring our results closer to these earlier works (as minimax equilibria can be shown to be neutrally stable).

\section*{Acknowledgments}

I would like to thank Spyros Alexakis for significantly helping me with the linear stability analysis of selfish routing (in Section \ref{linear_analysis}) some seven years ago. I would also like to thank the anonymous reviewers of a related submission whose comments helped improve the quality of this paper. A (very) preliminary version of results related to this paper (based on simulation rather than rigorous theoretical analysis) appeared in a paper I coauthored with Rob Schapire and Jen Rexford in a workshop on tackling systems problems with machine learning techniques (SysML) in 2008. Part of the work in this paper was performed during my stint as a senior researcher at Deutsche Telekom Laboratories (T-Labs) in Berlin, Germany between 2008 and 2011. I was introduced to multiplicative weights updates during an internship at Akamai Technologies in the summer of 2000 where the idea was coined that they may well serve as an optimization primitive.

\bibliographystyle{abbrvnat}
\bibliography{real}

\begin{thebibliography}{46}
\providecommand{\natexlab}[1]{#1}
\providecommand{\url}[1]{\texttt{#1}}
\expandafter\ifx\csname urlstyle\endcsname\relax
  \providecommand{\doi}[1]{doi: #1}\else
  \providecommand{\doi}{doi: \begingroup \urlstyle{rm}\Url}\fi

\bibitem[Absil and Kurdyka(2006)]{Absil}
P.-A. Absil and K.~Kurdyka.
\newblock On the stable equilibrium points of gradient systems.
\newblock \emph{Systems and Control Letters}, 55\penalty0 (7):\penalty0
  573--577, 2006.

\bibitem[Arora et~al.(2012)Arora, Hazan, and Kale]{AHK}
S.~Arora, E.~Hazan, and S.~Kale.
\newblock The multiplicative weights update method: A meta-algorithm and its
  applications.
\newblock \emph{Theory of Computing}, 8:\penalty0 121--164, 2012.

\bibitem[Beckmann et~al.(1956)Beckmann, McGuire, and
  Winsten]{Beckmann-McGuire-Winsten}
M.~Beckmann, C.~McGuire, and C.~Winsten.
\newblock \emph{Studies in the Economics of Transportation}.
\newblock Yale University Press, 1956.

\bibitem[Bertsekas and Tsitsiklis(1989)]{Bertsekas-Tsitsiklis}
D.~P. Bertsekas and J.~Tsitsiklis.
\newblock \emph{Parallel and Distributed Computation: Numerical Methods}.
\newblock Prentice Hall, 1989.

\bibitem[Bomze(1998)]{Bomze}
I.~M. Bomze.
\newblock On standard quadratic optimization problems.
\newblock \emph{Journal of Global Optimization}, 13:\penalty0 369--387, 1998.

\bibitem[Bomze et~al.(2008)Bomze, Locatelli, and Tardella]{BLT}
I.~M. Bomze, M.~Locatelli, and F.~Tardella.
\newblock New and old bounds for standard quadratic optimization: Dominance,
  equivalence and comparability.
\newblock \emph{Mathematical Programming}, 115:\penalty0 31--64, 2008.

\bibitem[Chastain et~al.(2014)Chastain, Livnat, Papadimitriou, and
  Vazirani]{MWUA}
E.~Chastain, A.~Livnat, C.~Papadimitriou, and U.~Vazirani.
\newblock Algorithms, games, and evolution.
\newblock \emph{PNAS}, 111\penalty0 (29):\penalty0 10620--10623, 2014.

\bibitem[Chong and Zak(2008)]{Zak}
E.~K.~P. Chong and S.~H. Zak.
\newblock \emph{An introduction to optimization}.
\newblock John Wiley \& Sons, Hoboken, New Jersey, third edition, 2008.

\bibitem[Conitzer(2013)]{Conitzer}
V.~Conitzer.
\newblock The exact computational complexity of evolutionarily stable
  strategies.
\newblock In \emph{Proc. WINE 2013}, pages 96--108, 2013.

\bibitem[Crespi et~al.(2005)Crespi, Ginchev, and Rocca]{Crespi}
G.~P. Crespi, I.~Ginchev, and M.~Rocca.
\newblock Existence of solutions and star-shapedness in {M}inty variational
  inequalities.
\newblock \emph{Journal of Global Optimization}, 32:\penalty0 485--494, 2005.

\bibitem[Daniilidis and Hadjisavvas(1999)]{DanHad}
A.~Daniilidis and N.~Hadjisavvas.
\newblock Characterization of nonsmooth semistrictly quasiconvex and strictly
  quasiconvex functions.
\newblock \emph{Journal of Optimization Theory and Applications}, 102\penalty0
  (3):\penalty0 525--536, 1999.

\bibitem[Etessami and Lochbihler(2008)]{Etessami}
K.~Etessami and A.~Lochbihler.
\newblock The computational complexity of evolutionarily stable strategies.
\newblock \emph{Int J Game theory}, 37:\penalty0 93--103, 2008.

\bibitem[Fischer and Voecking(2004)]{SelfishRoutingEvolution}
S.~Fischer and B.~Voecking.
\newblock On the evolution of selfish routing.
\newblock In \emph{Proc. European Symposium on Algorithms (ESA)}, pages
  323--334, 2004.

\bibitem[Fischer et~al.(2010)Fischer, Raecke, and
  Voecking]{FastConvergence-Journal}
S.~Fischer, H.~Raecke, and B.~Voecking.
\newblock Fast convergence to {W}ardrop equilibria by adaptive sampling
  methods.
\newblock \emph{SIAM Journal on Computing}, 39\penalty0 (8):\penalty0
  3700--3735, 2010.

\bibitem[Freund and Schapire(1997)]{FreundSchapire1}
Y.~Freund and R.~E. Schapire.
\newblock A decision-theoretic generalization of on-line learning and an
  application to boosting.
\newblock \emph{Journal of Computer and System Sciences}, 55\penalty0
  (1):\penalty0 119--139, 1997.

\bibitem[Freund and Schapire(1999)]{FreundSchapire2}
Y.~Freund and R.~E. Schapire.
\newblock Adaptive game playing using multiplicative weights.
\newblock \emph{Games and Economic Behavior}, 29:\penalty0 79--103, 1999.

\bibitem[Hofbauer and Sandholm(2009)]{StableGames}
J.~Hofbauer and W.~Sandholm.
\newblock Stable games and their dynamics.
\newblock \emph{Journal of Economic Theory}, 144:\penalty0 1665--1693, 2009.

\bibitem[Hofbauer and Sigmund(1988)]{Hofbauer-Sigmund}
J.~Hofbauer and K.~Sigmund.
\newblock \emph{The theory of evolution and dynamical systems}.
\newblock Cambridge University Press, Cambridge, 1988.

\bibitem[Hofbauer et~al.(1979)Hofbauer, Schuster, and Sigmund]{HSS}
J.~Hofbauer, P.~Schuster, and K.~Sigmund.
\newblock A note of evolutionary stable strategies and game dynamics.
\newblock \emph{J. theor. Biology}, 81:\penalty0 609--612, 1979.

\bibitem[Horn and Johnson(1985)]{Horn}
R.~Horn and C.~Johnson.
\newblock \emph{Matrix Analysis}.
\newblock Cambridge, 1985.

\bibitem[John(1998)]{John1}
R.~John.
\newblock Variational inequalities and pseudomonotone functions: Some
  characterizations.
\newblock In J.~P.~C. et~al., editor, \emph{Generalized Convexity, Generalized
  Monotonicity: Recent Results}, pages 291--301. Kluwer Academic Publishers,
  1998.

\bibitem[John(2001)]{John2}
R.~John.
\newblock A note on {M}inty variational inequalities and generalized
  monotonicity.
\newblock In N.~Hadjisavvas, J.~E. Martinez-Legaz, and J.-P. Penot, editors,
  \emph{Generalized Convexity and Generalized Monotonicity}. Springer-Verlag,
  2001.

\bibitem[Kale(2007)]{Kale}
S.~Kale.
\newblock \emph{Efficient Algorithms Using the Multiplicative Weights Update
  Method}.
\newblock PhD thesis, Princeton University, 2007.

\bibitem[Kleinberg et~al.(2009)Kleinberg, Piliouras, and Tardos]{Piliouras}
R.~Kleinberg, G.~Piliouras, and E.~Tardos.
\newblock Multiplicative updates outperform generic no-regret learning in
  congestion games.
\newblock In \emph{Proc. STOC}, May/Jun. 2009.

\bibitem[Konnov(1998)]{Konnov}
I.~V. Konnov.
\newblock On quasimonotone variational inequalities.
\newblock \emph{Journal of Optimization Theory and Applications}, 99\penalty0
  (1):\penalty0 165--181, 1998.

\bibitem[Koutsoupias and Papadimitriou(2009)]{PriceOfAnarchy}
E.~Koutsoupias and C.~Papadimitriou.
\newblock Worst-case equilibria.
\newblock \emph{Computer Science Review}, 3\penalty0 (2):\penalty0 65--69,
  2009.

\bibitem[LaSalle(1986)]{LaSalle}
J.~P. LaSalle.
\newblock \emph{The Stability and Control of Discrete Processes}.
\newblock Springer-Verlag, New York, 1986.

\bibitem[Li and Yorke(1975)]{PeriodThree}
T.~Y. Li and J.~A. Yorke.
\newblock Period three implies chaos.
\newblock \emph{American Mathematical Monthly}, 82:\penalty0 985--992, 1975.

\bibitem[Littlestone and Warmuth(1994)]{Littlestone}
N.~Littlestone and M.~K. Warmuth.
\newblock The weighted majority algorithm.
\newblock \emph{Information and Computation}, 108:\penalty0 212--261, 1994.

\bibitem[{Maynard Smith}(1982)]{Evolution}
J.~{Maynard Smith}.
\newblock \emph{Evolution and the Theory of Games}.
\newblock Cambridge University Press, 1982.

\bibitem[{Maynard Smith} and Price(1973)]{TheLogicOfAnimalConflict}
J.~{Maynard Smith} and G.~R. Price.
\newblock The logic of animal conflict.
\newblock \emph{Nature}, 246:\penalty0 15--18, 1973.

\bibitem[Mehta et~al.(2015)Mehta, Panages, and Piliouras]{MPP}
R.~Mehta, I.~Panages, and G.~Piliouras.
\newblock Natural selection as an inhibitor of genetic diversity:
  Multiplicative weights updates algorithm and a conjecture of haploid
  genetics.
\newblock In \emph{Proc. ITCS}, 2015.

\bibitem[Nachbar(1990)]{Nachbar}
J.~H. Nachbar.
\newblock {``Evolutionary''} selection dynamics in games: Convergence and limit
  properties.
\newblock \emph{International Journal of Game Theory}, 19:\penalty0 59--89,
  1990.

\bibitem[Nisan(2006)]{Nisan-ESS}
N.~Nisan.
\newblock A note on the computational hardness of evolutionary stable
  strategies.
\newblock Report no. 76, Electronic Colloquium of Computational Complexity,
  2006.

\bibitem[Nisan et~al.(2007)Nisan, Roughgarden, Tardos, and
  Vazirani]{Algorithmic-Game-Theory}
N.~Nisan, T.~Roughgarden, E.~Tardos, and V.~Vazirani, editors.
\newblock \emph{Algorithmic Game Theory}.
\newblock Cambridge University Press, 2007.

\bibitem[Rexford and Dovrolis(2010)]{Clean-Slate}
J.~Rexford and C.~Dovrolis.
\newblock Future {I}nternet architecture: Clean-slate versus evolutionary
  research.
\newblock \emph{Communications of the ACM}, 53\penalty0 (9), 2010.

\bibitem[Roughgarden(2005)]{Roughgarden}
T.~Roughgarden.
\newblock \emph{Selfish Routing and the Price of Anarchy}.
\newblock MIT Press, 2005.

\bibitem[Roughgarden and Tardos(2002)]{HowBadisSelfishRouting}
T.~Roughgarden and E.~Tardos.
\newblock How bad is selfish routing?
\newblock \emph{Journal of the ACM}, 49\penalty0 (2):\penalty0 236--259, 2002.

\bibitem[Sandholm(2010)]{PopulationGames}
W.~H. Sandholm.
\newblock \emph{Population Games and Evolutionary Dynamics}.
\newblock MIT Press, 2010.

\bibitem[Steele(2004)]{Steele}
J.~M. Steele.
\newblock \emph{The Cauchy-Schwarz Master Class}.
\newblock Cambridge University Press, New York, 2004.

\bibitem[Taylor and Jonker(1978)]{TaylorJonker}
P.~Taylor and L.~Jonker.
\newblock Evolutionary stable strategies and game dynamics.
\newblock \emph{Mathematical Biosciences}, 16:\penalty0 76--83, 1978.

\bibitem[Thomas and Pohley(1982)]{TP}
B.~Thomas and H.-J. Pohley.
\newblock On a global representation of the dynamical characteristics in
  {ESS}-models.
\newblock \emph{BioSystems}, 15:\penalty0 141--153, 1982.

\bibitem[Wardrop(1952)]{Wardrop}
J.~Wardrop.
\newblock Some theoretical aspects of road traffic research.
\newblock In \emph{Proceedings of the Institute of Civil Engineers, Pt. II},
  volume~1, pages 325--378, 1952.

\bibitem[Weibull(1995)]{Weibull}
J.~W. Weibull.
\newblock \emph{Evolutionary Game Theory}.
\newblock MIT Press, 1995.

\bibitem[Weissing(1991)]{Weissing}
F.~J. Weissing.
\newblock Evolutionary stability and dynamic stability in a class of
  evolutionary normal form games.
\newblock In R.~Selten, editor, \emph{Game Equilibrium Models I. Evolution and
  Game Dynamics}, pages 29--97. Springer, Berlin, 1991.

\bibitem[Zeeman(1979)]{Zeeman}
E.~C. Zeeman.
\newblock Population dynamics from game theory.
\newblock \emph{Springer Lecture Notes in Mathematics}, 819:\penalty0 472--497,
  1979.

\end{thebibliography}

\end{document}